\documentclass[12pt]{amsart}
   	
\usepackage{graphicx}
\usepackage{color}

\newtheorem{theorem}{Theorem}
\newtheorem{lemma}[theorem]{Lemma}
\newtheorem{corollary}[theorem]{Corollary}

\title[Optimal Retirement Income Tontines]{Optimal Retirement Income Tontines}
\author[M.A. Milevsky and T.S. Salisbury]{Moshe A. Milevsky and Thomas S. Salisbury}
\thanks{Milevsky is an Associate Professor of Finance at the Schulich School of Business,
York University, and Executive Director of the IFID Centre. Salisbury is a Professor in the Department of Mathematics and Statistics at York University. The authors acknowledge funding from Netspar (Milevsky), a Schulich Research Fellowship (Milevsky) and from NSERC (Salisbury). They wish to thank Rejo Peter, Dajena Collaku, Simon Dabrowski, Alexandra Macqueen and Branislav Nikolic for research as well as editorial assistance. The contact author (Salisbury) can be reached at: salt@yorku.ca} 
\date{24 March 2015 (Final Version 4.0)}				

\begin{document}
\maketitle

\begin{abstract} 

Tontines were once a popular type of mortality-linked investment pool. They promised enormous rewards to the last survivors at the expense of those died early. And, while this design \emph{appealed to the gambling instinct}, it is a suboptimal way to generate retirement income. Indeed, actuarially-fair life annuities making constant payments -- where the insurance company is exposed to longevity risk -- induce greater lifetime utility. However, tontines do not have to be structured the historical way, i.e. with a constant cash flow shared amongst a shrinking group of survivors. Moreover, insurance companies do not sell actuarially-fair life annuities, in part due to aggregate longevity risk. 

We derive the tontine structure that maximizes lifetime utility. Technically speaking we solve the Euler-Lagrange equation and examine its sensitivity to (i.) the size of the tontine pool $n$, and (ii.) individual longevity risk aversion $\gamma$. We examine how the optimal tontine varies with $\gamma$ and $n$, and prove some qualitative theorems about the optimal payout. Interestingly, Lorenzo de Tonti's original structure is optimal in the limit as longevity risk aversion $\gamma \to \infty$. We define the \emph{natural tontine} as the function for which the payout declines in exact proportion to the survival probabilities, which we show is near-optimal for all $\gamma$ and $n$. We conclude by comparing the utility of optimal tontines to the utility of loaded life annuities under reasonable demographic and economic conditions and find that the life annuity's advantage over the optimal tontine is minimal. 

In sum, this paper's contribution is to (i.) rekindle a discussion about a retirement income product that has been long neglected, and (ii.) leverage economic theory as well as tools from mathematical finance to design the next generation of tontine annuities.

\end{abstract}

\newpage

\section{Introduction and Executive Summary}

As policymakers, academics and the public at large grow increasingly concerned about the cost of an aging society, we believe it is worthwhile to go back in time and examine the capital market instruments used to finance retirement in a period before social insurance, defined benefit (DB) pensions, and annuity companies. Indeed, in the latter part of the 17th century and for almost two centuries afterwards, one of the most popular ``Retirement Income'' investments in the developed world was not a stock, bond or a mutual fund (although they were available). In fact, the method used by many individuals to generate income in the senior years of the lifecycle was a so-called tontine scheme sponsored by government\footnote{Sources: Weir (1989), Poterba (2005) and McKeever (2009).}. Part annuity, part lottery and part hedge fund, the tontine -- which recently celebrated its 360th birthday -- offered a lifetime of income that increased as other members of the tontine pool died off and their money was distributed to survivors.  

The classical tontine investment pool is quite distinct from its public image as a lottery for centenarians in which the longest survivor wins all the money in a pool. In fact, the tontine is both more subtle and more elegant. For those readers who are not familiar with the historical tontine, here is a simple example. Imagine a group of 1,000 soon-to-be retirees who band together and pool \$1,000 each to purchase a million-dollar U.S. Treasury bond (with a very long, or even perpetual maturity date) paying 3\% coupons\footnote{We state it this way for simplicity of exposition. Strictly speaking, the purchase should be the coupon stream only, since the principal is never returned.}. The bond generates \$30,000 in interest yearly, which is split amongst the 1,000 participants in the pool, for a 30,000 / 1,000 = guaranteed \$30 dividend per member. A custodian holds the big bond -- taking no risk and requiring no capital -- and charges a trivial fee to administer the annual dividends. In a tontine scheme the members agree that if and when they die, their guaranteed \$30 dividend is split amongst those who still happen to be alive. 

For example, if one decade later only 800 original investors are alive, the \$30,000 coupon is divided into 800, for a \$37.50 dividend each. Of this, \$30 is the guaranteed dividend and \$7.50 is \emph{other people's money}. Then, if two decades later only 100 survive, the annual cash flow to survivors is \$300, which is a \$30 guaranteed dividend plus \$270. When only 30 remain, they each receive \$1,000 in dividends. The extra payments -- above and beyond the guaranteed \$30 dividend --  are the mortality credits. In fact, under this scheme payments are expected to increase at the rate of mortality, which is a type of super-inflation hedge. 

The tontine, of course, differs from a conventional life annuity. Although both offer income for life and pool longevity risk, the mechanics and therefore the cost to the investor are quite different.  The annuity (or its issuing insurance company) promises predictable guaranteed lifetime payments, but this comes at a cost  -- and regulator-imposed capital requirements -- that inevitably makes its way to the annuitant. Indeed, the actuaries make very conservative assumptions regarding how long annuitants are likely to live. In future (especially under the regulations of Solvency II) annuity-issuers might have to hold capital and reserves against aggregate longevity risk. In contrast, the tontine custodian divides the variable $X$ (bond coupons received) by the variable $Y$ (participants alive) and sends out cheques\footnote{In between the pure tontine and the conventional annuity lies the \emph{participating annuity} which shares some longevity risk within a pool using a long-term smoothing scheme. Unfortunately, the actual formula for determining the smoothing mechanism is anything but smooth and lacks transparency. That said, self-annuitization schemes (GSA), proposed by Piggott, Valdez and Detzel (2005) have been growing in academic popularity.}. It is cleaner to administer, less capital-intensive and -- in its traditional form -- results in an increasing payment stream over time (assuming, of course, that you are alive). 

Underlying this paper is the argument that properly-designed tontines should be on the menu of products available to individuals as they transition into their retirement years.

\subsection{History vs. the Future}

Lorenzo de Tonti was a colorful Italian banker who in the 1650s promoted the scheme which now shares his name. He described the tontine as a mixture of lottery and insurance equally driven by fear of old age, and economic greed.  Tontines were first introduced in Holland, then very successfully in France\footnote{Source: Jennings and Trout (1982).} and a few decades later, in 1693, England's King William III presided over the first government-issued tontine. In this first government-issued tontine, in exchange for a minimum payment of \pounds 100, participants were given the option of (a) investing in a tontine scheme paying guaranteed \pounds 10 dividends until the year 1700 and then \pounds 7 thereafter, or (b) a life annuity paying \pounds 14 for life, but with no survivorship benefits. In our view, the choice between two possible ways of financing retirement income is something that could conceivably be offered again in the future. Perhaps the choice between a tontine (in which longevity risk is pooled) and a life annuity (in which payments are guaranteed) is a paradigm for many of the choices retirees, and society, now face. 

Back to the 1693 tontine of King William, over a thousand Englishmen (but very few women) decided to invest in the tontine. The sums involved weren't trivial. The \pounds 100 entry fee would be worth close to \pounds 100,000 today\footnote{Source: Lewin (2003).}. So, this was no impulse lottery purchase. Rather, some investors picked the tontine because they wanted the skewness -- that is, the potential for a very large payout if they survived -- while others wanted the more predictable annuity income option. The oldest survivor of King William's tontine of 1693 lived to age 100, earning \pounds 1,000 in her final year -- and no doubt well tended-to by her family.\footnote{Source: Finlaison (1829).} See Milevsky (2014) for more on King William's tontine.

Our main practical points will be that (i) tontines and life annuities have co-existed in the past, so perhaps they can in the future; (ii) there is no reason to construct the tontine payout function such that the last survivor receives hundreds of multiples of their initial investment;  and (iii) a properly-constructed tontine can result in lifetime utility that is comparable to the utility of a life annuity. Indeed, both the tontine and annuity hedge individuals against their idiosyncratic mortality risk: The difference is that the tontine leaves individuals exposed to aggregate longevity risk, while the annuity transfers this risk to the vendor -- for a fee. If this fee is sufficiently high, the tontine will offer greater utility than the annuity.

\subsection{Outline of the Paper}

To our knowledge, we are the first to derive the properties of an optimized tontine payout structure, in contrast to the optimal life annuity which is well-known in the literature. We believe this is part of the ``optimal insurance design'' literature, similar to the recent paper by Bernard, et. al. (2015). An interesting insurance product, that lies between a tontine and annuity, is the \emph{pooled annuity}. Stamos (2008) analyzes the optimal design of such a pool, while Donnelly, Guillen and Nielsen (2014) compare its utility to that of alternatives and a recent article by Donnelly (2015) examines the ``fairness'' of pooled annuity funds in general. As with a tontine, pooled annuity funds insulate the issuer from longevity risk. These hybrid designs yield utility intermediate between that of a tontine and an annuity, at the expense of a more complex product design (from the client point of view). We will discuss this, and review the broader literature, in Section \ref{lit}. 

The remainder of this paper is organized as follows. Section \ref{1693} very briefly describes the highlights of the first ever English government tontine of 1693, which interestingly had some aspects of our \emph{optimal tontine.} It also provides some motivation and colour to the paper. Section \ref{theory} is the theoretical core of the paper, which derives the properties of the optimal tontine structure and compares this optimal structure with the life annuity.  Section \ref{numerics} uses and applies the results from Section \ref{theory} and addresses how 21st-century tontines might be constructed. Section \ref{conc} concludes the paper and Section \ref{apen} is an appendix with most proofs and derivations.

\section{The First English Tontine: Non-Flat Payout}
\label{1693}

In early 1693, during the fourth year of the reign of King William III and Queen Mary II, the English Parliament passed the so-called \emph{Million Act}, which was an attempt to raise one million pounds towards carrying on the war against France.  The \emph{Million Act} specified that prior to May 1, 1693, any British native or foreigner could purchase a tontine share from the Exchequer for \pounds 100 and thus gain entry into the first British government tontine scheme.

For \pounds 100 an investor could select any nominee of any age -- including the investor himself -- on whose life the tontine would be contingent. Dividend payments would be distributed to the investor as long as the nominee was still alive. Now, to put the magnitude of the minimal \pounds 100 investment in perspective, the average annual wage of building laborers in England during the latter part of the 17th century was approximately \pounds 16 and a few shillings per year\footnote{Source: Lewin (2003).} -- so the entry investment in the tontine pool far exceeded the average industrial wage, and the annual dividends alone might serve as a decent pension for a common labourer. Accordingly, it is quite plausible to argue that the 1693 tontine was an investment for the rich and perhaps even one of the first exclusive hedge funds.

This was a simpler structure compared to the original tontine scheme envisioned by Lorenzo de Tonti in the year 1653, which involved multiple classes. In the 1693 English tontine, each share of \pounds 100 would entitle the investor to an annual dividend of \pounds 10 for seven years (until June 1700), after which the dividends would be reduced to \pounds 7 per share. The 10\% and 7\% tontine dividend rate exceeded prevailing (market) interest rates in England at the end of the 17th century, which were officially capped at 6\%.\footnote{Source: Homer and Sylla (2005).} Note the declining structure of the interest payments, which is a preview of our soon-to-come discussion (in Section 3) about the optimal tontine payout function.

Moving on to the annuity side of the offering, to further entice investors to participate in the tontine scheme, the Act included a unique ``sweetener" or bonus provision. It stipulated that if the entire \pounds 1,000,000 target wasn't subscribed by May 1693 -- thus reducing the size of the lottery payoff for the final survivor -- the investors who had enrolled in the tontine during the six-month subscription period (starting in November 1692) would have the option of converting their \pounds 100 tontine shares into a life annuity paying \pounds 14 per year. Under this alternative, the 14\% dividend payments were structured as a conventional life annuity with no group survivorship benefits or tontine features. (Think of a single premium life annuity.) This ``option to choose'' between a tontine and a life annuity is quite intriguing and fodder for further research. 

Alas, the funds raised by early May 1963 fell far short of the million-pound target. According to records maintained by the Office of the Exchequer, stored within the Archives of the British Library\footnote{Source:  Dickson (1967).}, a total of only \pounds 377,000 was subscribed and approximately 3,750 people were nominated\footnote{Note that one nominee could have multiple shares of \pounds 100 contingent on their life.} to the tontine prior to May 1693. This then triggered the option to exchange the tontine into a 14\% life annuity.

Interestingly, it seems that a total of 1,013 nominees (representing 1,081 tontine shares) remained in the original 10\%/7\% tontine, while the other two-thirds elected to convert their tontine shares into a 14\% life annuity contingent on the same nominee. (See Milevsky (2014) for an extensive historical discussion of the 1693 tontine, the characteristics of investors who selected the tontine versus the annuity and some of the empirical factors driving the decisions made by participants. Here we only describe the highlights.)
 
At first glance, it is rather puzzling why anyone would stay in the tontine pool instead of switching to the life annuity. On a present-value basis, a cash flow of \pounds10 for seven years and \pounds 7 per annum thereafter is much less valuable compared to \pounds 14 per annum for life. As discussed in Milevsky (2014), the actuarial present value of the 10\%/7\% combination at the 6\% official interest rate was worth approximately \pounds 133 at the (typical nominee) age of ten, whereas the value of the life annuity was worth almost \pounds 185. (Remember, the original investment was \pounds 100.) It should come as no surprise that the British government was losing money on these 14\% annuities\footnote{Source: Finlaison (1829).} -- and they were offering these terms to anyone, regardless of how young they were. Why did anyone stay in the tontine pool?

Could differing views on mortality and/or risk aversion explain why some investors switched, while others didn't? Can a rational explanation for this seemingly irrational choice be found? The framework we will introduce in Section 3 might help shed light on this decision, or at least rationalize the choice.

Motivated by the historical tontine, in the next section we present an economic theory to describe and understand the choice theory illuminating who might elect to participate in a tontine and who might choose a life annuity, as well as the properties of a tontine that are likely to generate the highest lifetime utility. To pre-empt the result, we will show that although the life annuity clearly dominated the tontine in terms of lifetime utility, a declining tontine payout structure (as in the 1693 tontine) is in fact optimal. So, having a tontine that pays 10\% in the first few years (to the entire pool) and then 7\% for the remaining years (to the entire pool) is a historical aberration but will be shown to be based in sound economic theory.

\section{Tontine vs. Annuity: Economic Theory}
\label{theory}

We assume an objective survival function ${}_tp_x$, for an individual aged $x$ to survive $t$ years. One purpose of the tontine structure is to insulate the issuer from the risk of a stochastic (or simply mis-specified) survival function, but in this paper we assume ${}_tp_x$ is given, and applies to all individuals. We leave for future work the exploration of subjective vs. objective survival rates, or stochastic hazard rates. We assume that the tontine pays out continuously, as opposed to quarterly or monthly. For simplicity, we assume a constant risk-free interest rate $r$, though it would be simple to incorporate a term structure (all funds contributed are invested at time 0, so payouts could be funded from a static bond portfolio as easily as from a money market account). Our basic annuity involves annuitants (who are also the nominees) each paying \$1 to the insurer initially, and receiving in return an income stream of $c(t)\,dt$ for life. The constraint on these annuities is that they are fairly priced, in other words that with a sufficiently large client base, the initial payments invested at the risk-free rate will fund the called-for payments in perpetuity (later we discuss the implications of insurance loadings). This implies a constraint on the annuity payout function $c(t)$, namely that 
\begin{equation}
\label{annuityconstraint}
\int_0^\infty e^{-rt}{}_tp_x\, c(t)\,dt=1.
\end{equation}
Though $c(t)$ is the payout rate per survivor, note that the payout rate per initial dollar invested is ${}_tp_x\,c(t)$. We will return to this later. 

Letting $u(c)$ denote the instantaneous utility of consumption (a.k.a. the felicity function), a rational annuitant (with lifetime $\zeta$) having no bequest motive will choose a life annuity payout function for which $c(t)$ maximizes the discounted lifetime utility:
\begin{equation}
E[\int_0^\zeta e^{-rt}u(c(t))\,dt]=\int_0^\infty e^{-rt}{}_tp_x\, u(c(t))\,dt
\label{annuityutilityspecification}
\end{equation}
where $r$ is (also) the subjective discount rate (SDR), all subject to the constraint \eqref{annuityconstraint}. 

By the Euler-Lagrange theorem\footnote{Source: Elsgolc (2007) page 51 or Gelfand and Fomin (2000), page 15.}, this implies the existence of a constant $\lambda$ such that 
\begin{equation}
e^{-rt}{}_tp_x\,u'(c(t))=\lambda e^{-rt}{}_tp_x \quad\text{for every $t$.}
\label{lagrangecondition}
\end{equation}
In other words, $u'(c(t))=\lambda$ is constant, so provided that utility function $u(c)$ is strictly concave, the optimal annuity payout function $c(t)$ is also constant. That constant is now determined by \eqref{annuityconstraint}, showing the following:
\begin{theorem} Optimized life annuities have constant $c(t)\equiv c_0$, where 
$$
c_0=\Big[\int_0^\infty e^{-rt}{}_tp_x\,dt\Big]^{-1}.
$$ 
\end{theorem} 

This result can be traced back to Yaari (1965) who showed that the optimal (retirement) consumption profile is constant (flat) and that 100\% of wealth is annuitized when there is no bequest motive. For more details and an alternate proof, see the excellent book by Cannon and Tonks (2008) and specifically the discussion on annuity demand theory in chapter 7.

\subsection{Optimal Tontine Payout}
In practice, insurance companies funding the life annuity $c(t)$ are exposed to both systematic longevity risk (due to randomness or uncertainty in ${}_tp_x$), model risk (the risk that ${}_tp_x$ is mis-specified), as well as re-investment or interest rate risk (which is the uncertainty in $r$ over long horizons). The latter is not our focus here, so we will continue to assume that $r$ is a given constant for most of what follows. Note that even if re-investment rates were known with certainty, the insurance company would likely pay out less than the $c(t)$ implied by equation \eqref{annuityconstraint} as a result of required capital and reserves, effectively lowering the lifetime utility of the (annuity and the) retiree.

This brings us to the tontine structures we will consider as an alternative, in which a predetermined dollar amount is shared among survivors at every $t$. Let $d(t)$ be the rate at which funds are paid out per initial dollar invested, a.k.a. the tontine payout function. Our main point in this paper is that there is no reason for tontine payout function to be a constant fixed percentage of the initial dollar invested (e.g. 4\% or 7\%), as it was historically. In fact, we can pose the same question as considered above for annuities: what $d(t)$ is optimal for subscribers, subject to the constraint that sponsors of the tontine cannot sustain a loss? Note that the natural comparison is now between $d(t)$ and ${}_tp_x\,c(t)$, where $c(t)$ is the optimal annuity payout found above. 

Suppose there are initially $n$ subscribers to the tontine scheme, each depositing a dollar with the tontine sponsor. Let $N(t)$ be the random number of live subscribers at time $t$. Consider one of these subscribers. Given that this individual is alive, $N(t)-1\sim \text{Bin}(n-1,{}_tp_x)$. In other words, the number of other (live) subscribers at any time $t$ is binomially distributed with probability parameter ${}_tp_x$. 

So, as we found for the life annuity, this individual's discounted lifetime utility is
\begin{align*}
&E[\int_0^\zeta e^{-rt }u\Big(\frac{n d(t)}{N(t)}\Big)\,dt]=\int_0^\infty e^{-rt}{}_tp_x \,E[u\Big(\frac{n d(t)}{N(t)}\Big)\mid \zeta>t]\,dt\\
&\qquad=\int_0^\infty e^{-rt}{}_tp_x\sum_{k=0}^{n-1} \binom{n-1}{k}{}_tp_x^k(1-{}_tp_x)^{n-1-k}u\Big(\frac{nd(t)}{k+1}\Big)\,dt.
\end{align*}
The constraint on the tontine payout function $d(t)$ is that the initial deposit of $n$ should be sufficient to sustain withdrawals in perpetuity. Of course, at some point all subscribers will have died, so in fact the tontine sponsor will eventually be able to cease making payments, leaving a small remainder or windfall. But this time is not predetermined, so we treat that profit as an unavoidable feature of the tontine. Remember that we do not want to expose the sponsor to any longevity risk. It is the pool that bears this risk entirely. 

Our budget or pricing constraint is therefore that 
\begin{equation}
\label{tontineconstraint}
\int_0^\infty e^{-rt} d(t)\,dt=1.
\end{equation}

So, for example, if $d(t)=d_0$ is forced to be constant (the historical structure, which we call a \emph{flat tontine}), then the tontine payout function (rate) is simply $d_0=r$ (or somewhat more if a cap on permissible ages is imposed, replacing the upper bound of integration in \eqref{tontineconstraint} by a value less than infinity). We are instead searching for the \emph{optimal} $d(t)$ which we will find is far from constant.

By the Euler-Lagrange theorem from the Calculus of Variations, there is a constant $\lambda$ such that the optimal $d(t)$ satisfies 
\begin{equation}
e^{-rt}{}_tp_x\sum_{k=0}^{n-1} \binom{n-1}{k}{}_tp_x^k(1-{}_tp_x)^{n-1-k}\frac{n}{k+1}u'\Big(\frac{nd(t)}{k+1}\Big)=\lambda e^{-rt}
\label{opton1}
\end{equation}
for every $t$. Note that this expression directly links individual utility $u(\cdot)$ to the optimal participating annuity. Recall that a tontine is an extreme case of participation or pooling of all longevity risk. Equation \eqref{opton1} dictates exactly how a risk-averse retiree will trade off consumption against longevity risk. In other words, we are not advocating an \emph{ad hoc} actuarial process for smoothing realized mortality experience. 

Note that an actual mortality hazard rate $\mu_x$ does not appear in the above equation -- it appears only implicitly, in both ${}_tp_x$ and $\lambda$ (which is determined by \eqref{tontineconstraint}). Therefore, we will simplify our notation by re-parametrizing in terms of the probability: Let $D_u(p)$ satisfy 
\begin{equation}
p\sum_{k=0}^{n-1} \binom{n-1}{k}p^k(1-p)^{n-1-k}\frac{n}{k+1}u'\Big(\frac{nD_u(p)}{k+1}\Big)=\lambda.
\end{equation}
Substituting $p=1$ into the above equation, it collapses to $u'(D_u(1))=\lambda$. 

\begin{theorem}
The optimal tontine structure is $d(t)=D_u({}_tp_x)$, where $\lambda$ is chosen so \eqref{tontineconstraint} holds. 
\end{theorem}

It is feasible (but complicated) to solve this once a generic $u$ is given. But in the case of Constant Relative Risk Aversion (CRRA) utility it simplifies greatly. Let  $u(c)=c^{1-\gamma}/(1-\gamma)$ if $\gamma\neq 1$, and when $\gamma= 1$ take $u(c)=\log c$ instead. Define 
\begin{equation}
\theta_{n,\gamma}(p)=E\Big[\Big(\frac{n}{N(p)}\Big)^{1-\gamma}\Big]=\sum_{k=0}^{n-1} \binom{n-1}{k}p^{k}(1-p)^{n-1-k}\Big(\frac{n}{k+1}\Big)^{1-\gamma}
\end{equation}
where $N(p)-1\sim\text{Bin}(n-1,p)$. Set $\beta_{n,\gamma}(p)=p\theta_{n,\gamma}(p)$. Then 
\begin{corollary} 
\label{CRRAoptimum}
With CRRA utility, the optimal tontine has withdrawal rate
$D_{n,\gamma}^{\text{\rm OT}}(p)=D_{n,\gamma}^{\text{\rm OT}}(1)\beta_{n,\gamma}(p)^{1/\gamma}$, where
\begin{equation}
D_{n,\gamma}^{\text{\rm OT}}(1)=\Big[\int_0^\infty e^{-rt}\beta_{n,\gamma}({}_tp_x)^{1/\gamma}\,dt\Big]^{-1}.
\label{D(1)formula}
\end{equation}
\end{corollary}
\begin{proof}
Suppose $\gamma\neq 1$. Then the equation for $D_{n,\gamma}^{\text{OT}}(p)$  becomes that
\begin{equation}
D_{n,\gamma}^{\text{OT}}(p)^{-\gamma}p\theta_{n,\gamma}(p)=\lambda=D_{n,\gamma}^{OT}(1)^{-\gamma}.
\label{Dformula1}
\end{equation}
The constraint \eqref{tontineconstraint} now implies \eqref{D(1)formula}. A similar argument applies when $\gamma=1$. 
\end{proof}

\subsection{Payout Properties}
\label{properties}

It is worth emphasizing that $D_{n,\gamma}^{\text{OT}}(p)/D_{n,\gamma}^{\text{OT}}(1)=\beta_{n,\gamma}(p)^{1/\gamma}$ does not depend on the particular form of the mortality hazard rate $\mu_x$, or on the interest rate $r$, but only on the longevity risk aversion $\gamma$ and the number of initial subscribers to the tontine pool, $n$. In other words, the mortality hazard rate and $r$ enter into the expression for $D_{n,\gamma}^{\text{OT}}(p)$ only via the constant $D_{n,\gamma}^{\text{OT}}(1)$. 
We will prove the following in section \ref{proof1} of the appendix:
\begin{lemma} For any $n\ge 2$ and $\gamma>0$
\label{betabound}
\begin{enumerate}
\item $\beta_{n,\gamma}(p)$ increases with $p$;
\item $\beta_{n,\gamma}(0)=0$ and $\beta_{n,\gamma}(1)=1$;
\item $\beta_{n,1}(p)=p$,
\newline $\beta_{n,2}(p)=\frac{p}{n}(1+(n-1)p)$,
\newline $\beta_{n,3}(p)=\frac{p}{n^2}(1+3(n-1)p+(n-1)(n-2)p^2)$;
\item $\lim_{\gamma\to\infty}\beta_{n,\gamma}(p)=p^n$ and $\beta_{n,0}(p)=1-(1-p)^n$;
\item For $0<p<1$ we have that
$$
\text{$\beta_{n,\gamma}(p)$ is }
\begin{cases}
<p^\gamma, &\text{if $0<\gamma<1$}\\
>p^\gamma, &\text{if $1<\gamma$.}
\end{cases}
$$
\end{enumerate}
\end{lemma}

One immediate consequence of Corollary \ref{CRRAoptimum} and (a) of Lemma \ref{betabound} is that 
the optimal tontine payout rate $d(t)$ decreases with $t$. 

By (d) we have that $\lim_{\gamma\to\infty}\beta_{n,\gamma}(p)^{1/\gamma}=1$ for $p>0$, so by dominated convergence, $D_{n,\gamma}^{\text{OT}}(1)^{-1}=\int_0^\infty e^{-rt}\beta_{n,\gamma}({}_tp_x)^{1/\gamma}\,dt\to\frac1r$ as $\gamma\to\infty$. Therefore the (historical) flat tontine structure $D(p)=r$ is optimal in the limit as $\gamma\to\infty$. 

Likewise $\lim_{\gamma\downarrow 0}\beta_{n,\gamma}(p)^{1/\gamma}=0$ for $p<1$. Renormalizing, this implies that $D^{\text{OT}}_{n,\gamma}(p)$ concentrates increasingly near $p=1$. In other words, as we approach risk-neutrality, the optimal tontine comes closer and closer to exhausting itself immediately following $t=0$.

We define a {\it natural tontine} to have payout $d(t)=D_{\text{N}}({}_tp_x)$ where $D_{\text{N}}(p)$ is proportional to $p$, just as is the case for the annuity payment per initial dollar invested. Comparing the budget constraints \eqref{annuityconstraint} and \eqref{tontineconstraint}, we get that  $D_{\text{N}}(p)=pc_0$, so the natural tontine payout rate agrees with that of the annuity (which justifies our singling it out).  By Corollary \ref{CRRAoptimum} and (c) of Lemma \ref{betabound}, we see that the natural tontine is optimal for logarithmic utility $\gamma=1$. We will see in Section \ref{numerics} that the natural tontine is close to optimal when $n$ is large. We therefore propose the natural tontine as a reasonable structure for designing tontine products in practice, rather than expecting insurers to offer a range of products with differing $\gamma$'s. 

Figure \ref{fig5} shows the ratio of $D_{n,\gamma}^{\text{OT}}(p)$ to $D_{\text{N}}(p)$, with Gompertz hazard rate $\lambda(t)=\frac{1}{b}e^{\frac{x+t-m}{b}}$, and parametrized by $t$ via $p={}_tp_x$. This exponentially-rising hazard rate is the basic mortality assumption in much of the actuarial literature. The figure provides numerical evidence that higher risk aversion implies a preference for reducing consumption at early ages in favour of reserving funds to consume at advanced ages. 
To make this statement precise, define
$$
\Delta_{n,\gamma}(t)=\int_0^t e^{-rs}d(s)\,ds,
$$
namely the present value of payouts from the optimal tontine through time $t$, per initial dollar invested (or equivalently, the proportion of initial capital used to fund payouts till time $t$). We conjecture that
\begin{equation}
\gamma_1>\gamma_2\Rightarrow \Delta_{n,\gamma_1}(t) < \Delta_{n,\gamma_2}(t),
\label{depletioninequality}
\end{equation}
in any realistic situation (specifically, whenever the morality distribution has an increasing hazard rate). We have not yet succeeded in proving this conjecture, but instead will derive a number of partial results that support it. In fact, numerical evidence suggests that \eqref{depletioninequality} holds regardless of the mortality distribution, provided $\gamma_1$ is only of moderate size. It is in fact possible to construct pathological mortality distributions for which  \eqref{depletioninequality} fails when $\gamma_1$ and $\gamma_2$ are large. Specifically, this may actually happen during a low-hazard rate lull between two periods with high hazard rates. We discuss this at greater length in the appendix, where we will also prove the following instances of the above conjecture.
\begin{theorem}
\label{increasinggamma}
\eqref{depletioninequality} holds in the following situations:
\begin{enumerate}
\item For $0<\gamma_2\le 1\le \gamma_1$ and $t>0$ arbitrary;
\item For any fixed $\gamma_2$, $n$, and $t$, in the limit as $\gamma_1\to\infty$;
\item For some initial period of time $t\in(0,t_0)$, where $t_0$ depends on $\gamma_1>\gamma_2$ and $n$.
\end{enumerate}
\end{theorem}
We will also prove a modified version of \eqref{depletioninequality}, for fixed $\gamma_1>\gamma_2$ and $t$, in the limit as $n\to\infty$ (see Theorem \ref{asymptotictheorem}). 

In Section \ref{numerics} we will provide a variety of numerical examples that illustrate the optimal tontine payout function $d(t)$ as a function of longevity risk aversion 
$\gamma$ and the initial size of the tontine pool $n$.

\subsection{Optimal Tontine Utility vs. Annuity Utility} We now compare the annuity and tontine. 
\label{tontineannuitycagematch}
Let $U_{n,\gamma}^{\text{OT}}$ denote the utility of the optimal tontine. To compute this, suppose $\gamma\neq 1$, and observe that
\begin{equation}
\frac{D_{n,\gamma}^{\text{OT}}(p)^{1-\gamma}}{1-\gamma}p\theta_{n,\gamma}(p)=\frac{D_{n,\gamma}^{\text{OT}}(p)}{1-\gamma}D_{n,\gamma}^{\text{OT}}(p)^{-\gamma}p\theta_{n,\gamma}(p)
=\frac{D_{n,\gamma}^{\text{OT}}(p)}{1-\gamma}D_{n,\gamma}^{\text{OT}}(1)^{-\gamma}
\end{equation}
by \eqref{Dformula1}. The utility of the optimal tontine is therefore precisely
\begin{align*}
U_{n,\gamma}^{\text{OT}}
&=\int_0^\infty e^{-rt}\frac{D_{n,\gamma}^{\text{OT}}({}_tp_x)^{1-\gamma}}{1-\gamma}{}_tp_x\,\theta_{n,\gamma}({}_tp_x)\,dt 
= \frac{D_{n,\gamma}^{\text{OT}}(1)^{-\gamma}}{1-\gamma}\int_0^\infty e^{-rt}D_{n,\gamma}^{\text{OT}}({}_tp_x)\,dt\\
&=\frac{D_{n,\gamma}^{\text{OT}}(1)^{-\gamma}}{1-\gamma}
=\frac{1}{1-\gamma}\Big(\int_0^\infty e^{-rt}\beta_{n,\gamma}({}_tp_x)^{1/\gamma}\,dt\Big)^\gamma
\end{align*}
by \eqref{tontineconstraint} and \eqref{D(1)formula}. 

Consider instead the utility $U_\gamma^{\text{A}}$ provided by the annuity, namely 
\begin{equation}
U_\gamma^{\text{A}}=\int_0^\infty e^{-rt}{}_tp_x \frac{c_0^{1-\gamma}}{1-\gamma}\,dt
=\frac{\int_0^\infty e^{-rt}{}_tp_x\,dt}{(1-\gamma)\Big(\int_0^\infty e^{-rt}{}_tp_x\,dt\Big)^{1-\gamma}}
=\frac{1}{1-\gamma}\Big(\int_0^\infty e^{-rt}{}_tp_x\,dt\Big)^\gamma.
\end{equation}

\begin{theorem}
$U_{n,\gamma}^{\text{\rm OT}}<U_\gamma^{\text{\rm A}}$ for any $n$ and $\gamma>0$.
\label{utilityinequality}
\end{theorem}
\begin{proof}
For $\gamma\neq 1$ this follows from (e) of Lemma \ref{betabound} and the calculations given above. We show the case $\gamma=1$ in the appendix. Of course, the conclusion is economically plausible, as the tontine retains risk that the annuity doesn't.\end{proof}

\subsection{Indifference Annuity Loading}
The insurer offering an annuity will be modelled as setting aside some fraction of the initial deposits to fund the annuity's costs.  In other words, a fraction $\delta$ of the initial deposits are deducted initially, to fund risk management, capital reserves, etc. The balance, invested at the risk-free rate $r$ will fund the annuity. Therefore, with loading, \eqref{annuityconstraint} becomes that $\int_0^\infty e^{-rt}{}_tp_x\, c(t)\,dt=1-\delta$, which implies that $c(t)\equiv c_1=(1-\delta)c_0$ is the optimal payout structure for the annuity. The utility of the loaded annuity is therefore 
$$
U^{\text{LA}}_{\gamma,\delta}=
\int_0^\infty e^{-rt}{}_tp_x \frac{c_1^{1-\gamma}}{1-\gamma}\,dt
=\frac{(1-\delta)^{1-\gamma}\int_0^\infty e^{-rt}{}_tp_x\,dt}{(1-\gamma)\Big(\int_0^\infty e^{-rt}{}_tp_x\,dt\Big)^{1-\gamma}}
=\frac{c_0^{-\gamma}}{1-\gamma}(1-\delta)^{1-\gamma}
$$
for $\gamma\neq 1$, and $\frac{\log(c_0)+\log(1-\delta)}{c_0}$ for $\gamma=1$. 

In Section \ref{numerics} we will consider numerically the \emph{indifference loading} $\delta$ that, when applied to the annuity, makes an individual indifferent between the annuity and a tontine, i.e. 
$U^{\text{LA}}_{\gamma,\delta}=U_{n,\gamma}^{\text{OT}}$. 
It turns out that the loading $\delta$ decreases with $n$, in such a way that the \emph{total} loading $n\delta$ stays roughly stable. In other words, there is at most a fixed amount (roughly) that the insurer can deduct from the \emph{aggregate} annuity pool, regardless of the number of participants, before individuals start to derive greater utility from the tontine. We will illustrate this observation, at least for $1<\gamma\le 2$, by proving the following inequality in the appendix:
\begin{theorem}
\label{loadinginequality}
Suppose that $1<\gamma\le 2$. Then 
$\delta<\frac{1}{n}\big(\frac{c_0}{r}-1\big)$.
\end{theorem}
Note that $c_0>r$, since $c_0^{-1}=\int_0^\infty e^{-rt}{}_tp_x\,dt<\int_0^\infty e^{-rt}\,dt=r^{-1}$.

\subsection{Additional Considerations}

There are several issues we do not address here which we openly acknowledge and leave for future research. In particular, there is the role of \emph{credit risk} as well as the impact of \emph{stochastic mortality} which will change the dynamics between tontines and annuities. The existence of credit risk is a much greater concern for the buyers of life annuities, vs. tontines, given the risk assumed by the insurance sponsor. Likewise, under a stochastic mortality model the tontine payout would be more variable and uncertain, which might reduce the utility of the tontine relative to the life annuity. On the other hand, the capital charges added to the market price of the life annuity would be higher in a stochastic mortality model. We appreciate that these are an unanswered (or unaddressed) questions and leave the examination of the robustness of the {\it natural tontine} payout in a full stochastic mortality environment to a subsequent paper.  

We note that there is also the question of asymmetric mortality, in which the individual believes that his or her (subjective) hazard rate is less than that of the typical tontine investor (for whom the tontine payout function is optimized). Indeed, this is potentially the explanation for the significant fraction of investors in King William's tontine, who did not exercise their option to convert to an annuity. They (might have) thought their nominee was \emph{much} healthier than everyone else, and the only way to capitalize on that was to invest in the tontine. This leads to the idea that the {\it mortality loading}, or level of asymmetric mortality, would induce an individual to prefer the tontine to an annuity. This, obviously, ties into the issue of stochastic mortality and is being addressed in a follow-up paper by Milevsky and Salisbury (2015).

In the next section, we focus on the numerical implications of our current model.

\section{Numerical Examples for the Optimal Tontine}
\label{numerics}

Figure \ref{fig1} displays the range of possible 4\% flat tontine dividends over time, assuming an initial pool of $n=400$ nominees, under a Gompertz law of mortality with parameters $m=88.721$ and $b=10$. This mortality basis corresponds to a survival probability of ${}_{35}p_{65}=0.05$, i.e. from age 65 to age 100 and is the baseline value for several of our numerical examples. The figure clearly shows an increasing payment stream conditional on survival, which isn't the optimal function. 

\begin{center}
{\bf Figure \ref{fig1} goes here}
\end{center}

Indeed, in such a (traditional, historical) tontine scheme, the initial expected payment is quite low, relative to what a life annuity might offer in the early years of retirement; while the payment in the final years -- for those fortunate enough to survive -- would be both very high and quite variable. It is not surprising then that this form of tontine is both suboptimal from an economic utility point of view and also isn't very appealing to individuals (with reasonable risk aversions) who want to maximize their standard of living over their entire retirement lifespan.

\begin{center}
{\bf Figure \ref{fig2} goes here}
\end{center}

In contrast to Figure \ref{fig1} which displays the sub-optimal flat tontine, Figure \ref{fig2} displays the range of outcomes from the optimal tontine payout function, under the same interest rate and mortality basis. To be very specific, Figure \ref{fig2} is computed by solving for the value of $D_{n,\gamma}^{\text{OT}}(1)$ and then constructing $D_{n,\gamma}^{\text{OT}}(_tp_x)$ for $n=400, r=0.04$ and $\gamma=1$.  Once the payout function is known for all $t$, the number of survivors at the 10th and 90th percentile of the binomial distribution is used to bracket the range of the payout from age 65 to age 100. Clearly, the expected payout per survivor is relatively constant over the retirement years, which is much more appealing intuitively. Moreover, the discounted expected utility from a tontine payout such as the one displayed in Figure \ref{fig2} is much higher than the utility of the one displayed in Figure \ref{fig1}. 

\begin{center}
{\bf Table \ref{table06} goes here}
\end{center}

Table \ref{table06} displays the optimal tontine payout function for a very small pool of size $n=25$. These values correspond to the $D_{n,\gamma}^{\text{OT}}(_tp_x)$ values derived in section \ref{theory}. Notice how the optimal tontine payout function is quite similar (identical in the first significant digit) regardless of the individual's Longevity Risk Aversion (LoRA) $\gamma$, even when the tontine pool is relatively small at $n=25$. The minimum guaranteed dividend starts off at about 7\% at age 65 and then declines to approximately 1\% at age 95. Of course, the actual cash flow payout to an individual, conditional on being alive does not necessarily decline and in reality stays relatively constant. 

\begin{center}
{\bf Table \ref{table07} goes here}
\end{center}

Table \ref{table07} displays utility indifference loadings $\delta$ for a participant at age 60. Notice how even a retiree with a very high level of Longevity Risk Aversion (LoRA) $\gamma$, will select a tontine (with pool size $n \geq 20$) instead of a life annuity if the insurance loading is greater than 7.5\%; For less extreme levels of risk aversion, and larger pools, the required loading is much smaller (e.g. 30 b.p.). Recall that this is a one-time charge at the time of initiation, not an annualized value over the duration of the contract (which would be a significantly smaller value). 

\begin{center}
{\bf Table \ref{table08} goes here}
\end{center}

Table \ref{table08} computes certainty equivalent factors associated with natural tontines. See Section \ref{moreoncertaintyequiv} for further discussion of this comparison. If an individual with LoRA $\gamma \neq 1$ is faced with a tontine structure that is only optimal for someone with LoRA $\gamma=1$ (i.e. logarithmic utility) the welfare loss is minuscule. This is one reason we advocate the \emph{natural tontine} payout function, which is only precisely optimal for $\gamma=1$, as the basis for 21st-century tontines. 

\begin{center}
{\bf Figure \ref{fig3} goes here}
\end{center}

Figure \ref{fig3} compares the optimal tontine payout function $d(t)$ for different levels of Longevity Risk Aversion $\gamma$, and shows that the difference is barely noticeable when the tontine pool size is greater than $n=250$, mainly due to the effect of the law of large numbers. This curve traces the \emph{minimum} dividend that a survivor can expect to receive at various ages (i.e. when all purchasers survive). The median is (obviously) much higher (see Figure \ref{fig2}). 
\begin{center}
{\bf Figure \ref{fig4} goes here}
\end{center}

Figure \ref{fig4} illustrates that the optimal tontine payout function for someone with logarithmic $\gamma=1$ utility starts off paying the exact same rate as a life annuity regardless of the number of participants in the tontine pool, ie $d(0)=c_0$. But, for higher levels of longevity risk aversion $\gamma$ and a relatively smaller tontine pool, the function starts off at a lower value and declines at a slower rate. The perpetuity curve corresponds to $\gamma=\infty$.

\begin{center}
{\bf Figure \ref{fig5} goes here}
\end{center}

Figure \ref{fig5} shows that if we compare two retirees, the one who is more averse to longevity risk will prefer a higher Guaranteed Minimum Payout Rate (GMPR) at advanced ages. In exchange they will accept a lower GMPR at younger ages. We choose to quantify this by taking the natural tontine payout as our baseline, since this is the product structure we advocate. We take quite a small pool size here ($n=25$) to illustrate the effect. The effect persists at larger pool sizes, but is much less dramatic.

\begin{center}
{\bf Figure \ref{fig6} goes here}
\end{center}

Figure \ref{fig6} shows an actuarially-fair life annuity that guarantees 7.5\% for life starting at age 65. It provides more utility than an optimal tontine regardless of Longevity Risk Aversion (LoRA) or the size of the tontine pool. But, once an insurance loading is included, driving the annuity yield under the initial payout from the optimal tontine, the utility of the life annuity might be lower. The indifference loading is $\delta$ and is reported in Table \ref{table07}.

So here is our main takeaway and idea in the paper, once again. The historical tontine in which dividends to the entire pool are a constant (e.g. 4\%) interest rate over the entire retirement horizon are suboptimal because they create an increasing consumption profile that is both variable and undesirable. However, a tontine scheme in which interest payments to the pool early on are higher (e.g. 8\%) and then decline over time, so that the few winning centenarians receive a much lower interest rate (e.g. 1\%) is in fact the optimal design. Coincidently, King William's 1693 tontine had a similar declining structure of interest payments to the pool, which was quite rare historically. 

We are careful to distinguish between the guaranteed \emph{interest} rate (e.g. 8\% or 1\%) paid to the entire pool, and the expected \emph{dividend} to the individual investor in the optimal tontine, which will be relatively constant over time, as is evident from Figure \ref{fig2}. Of course, the present value of the interest paid to the entire pool over time is exactly equal to the original contribution made by the pool itself. We are simply re-arranging and parsing cash flows of identical present value, in a different manner over time.

We have shown that a tontine provides less utility than an actuarially-fair life annuity, which is reasonable given that the tontine exposes investors to longevity risk. What is striking is that the utility difference from a properly-designed tontine scheme is actually quite small when compared to an actuarially-fair life annuity, which is the workhorse of the pension economics and lifecycle literature. In fact, the utility from a tontine might actually be higher than the utility generated by a pure life annuity when the insurance (commission, capital cost, etc.) one-time loading exceeds 10\%. This result should not negate or be viewed as conflicting with the wide-ranging annuity literature which proves the optimality of life annuities in a lifecycle model. Both tontines and annuities hedge idiosyncratic mortality risk. In fact, what we show is that it is still optimal to fully hedge the remaining systemic longevity risk, but the instrument that one uses to do so depends on the relative costs. In other words, \emph{sharing} longevity risk amongst a relatively small ($n \leq 100$) pool of people doesn't create large dis-utilities or welfare losses, at least within a classical rational model. This finding can also be viewed as a further endorsement of the participating life annuity, which lies in between the tontine and the conventional life annuity.

\section{ Literature Review}
\label{lit}

This is not the place -- nor do we have the space -- for a full review of the literature on tontines, so we provide a selected list of key articles for those interested in further reading. For literature and all sources available for the 1693 tontine, we refer to the earlier mentioned historical paper by Milevsky (2014) as well as the book by Milevsky (2015).

More generally, the original tontine proposal by Lorenzo de Tonti appears in French in Tonti (1654) (and in English translation in the cited reference). The review article by Kopf (1927) and the book by O'Donnell (1936) are quite dated, but document how the historical tontine operated, discussing its checkered history, and providing a readable biography of some of its earliest promoters in Denmark, Holland, France and England. The monograph by Cooper (1972) is devoted entirely to tontines and the foundations of the 19th-century (U.S.) tontine insurance industry, which is based on the tontine concept but is somewhat different because of the savings and lapsation component. In a widely-cited article, Ransom and Sutch (1987) provide the background and story of the banning of tontine insurance in New York State, and then eventually the entire U.S. The comprehensive monograph by Jennings and Trout (1982) reviews the history of tontines, with particular emphasis on the French period, while carefully documenting payout rates and yields from most known tontines. For those interested in the pricing of mortality-contingent claims during the 17th and 18th century, as well as the history and background of the people involved, see Alter (1986), Poitras (2000), Hald (2003), Poterba (2005), Ciecka (2008a, 2008b), Rothschild (2009) as well as Bellhouse (2011), and of course, Homer and Sylla (2005) for the relevant interest rates. 

More recently, the newspaper article by Chancellor (2001), the book by Lewin (2003) and especially the recent review by McKeever (2009) all provide a very good history of tontines and discuss the possibility of a tontine revival. The standard actuarial textbooks, such as Promislow (2011) or Pitacco, et al. (2009) for example, each have a few pages devoted to the historical tontine principal. 

More relevant to the modern design of participating annuities and sharing longevity risk in the 21st century, a series of recent papers on pooled annuity funds have resurrected the topic. For example Piggot, Valdez and Detzel (2005), Valdez, Piggott and Wang (2006), Stamos (2008), Richter and Weber (2011), Donnelly, Guillen and Nielsen (2013) as well as Qiao and Sherris (2013) have all attempted to reintroduce and analyze tontine-like structures in one form or another. In these proposed structures, investors contribute to an initial pool of capital, which is invested, and then paid out over time to the survivors. It need not be the case that investments are risk free (as we have assumed for the tontines we analyze), so the control variables are the asset allocation and payout strategy. The prospectus given to pooled annuity investors should therefore specify both.

Donnelly, Guillen and Nielsen (2013) consider a different approach, in the context of pooled annuities, to what we above have called the indifference annuity loading. Instead of the loading being payed in a lump sum at the time of purchase (as we have done), they incorporate a variable fee into their guaranteed product, paid continuously over time. They find that this fee may be structured so the utility of the pooled annuity precisely matches that of the guaranteed product it is being compared to, at all times. They term this fee structure the ``lifetime breakeven costs'', and explore its properties and implications for consumption.

To further compare with our results, let us assume for now that the pooled annuity funds are indeed invested risk-free. Following the argument of Stamos (2008), one would then obtain the optimal \emph{proportion} of the remaining funds that should be paid out at each time, as a function of both time \emph{and the number of individuals remaining in the pool}. In a tontine, the amount withdrawn may vary with time, but not with the number of survivors. In other words, the prospectus for a tontine should lay out a specific dollar amount to be withdrawn in each year of the contract, so any uncertainty in the amount paid to an individual derives simply from the size of the pool with whom the withdrawal is split. A pooled annuity has the latter uncertainty, but in addition has uncertainty about the dollar amount to be withdrawn, which will vary in a path-dependent way according to how the mortality experience of the pool evolves. To know how much you will receive in year 10, it is no longer enough to know how many survivors there are in year 10, but you also need to know how many survived in each of years one through nine. 

There are both advantages and disadvantages to pooled annuities versus tontines. On the plus side, the pooled annuity has an extra variable to control for, so should yield higher utility -- with risk-free investing we expect its utility to be intermediate between that of a tontine and an annuity. As we have seen, the latter are in fact very close, so the utility gain (though real) is actually modest. On the down side, the path-dependent nature of a pooled annuity (or the annuity overlay introduced by Donnelly et. al. (2014)) makes it more complex to explain to participants, and more difficult for those participants to evaluate the risks and uncertainties of their income stream. For these reasons, we (on balance) advocate tontines. The simpler design of tontines certainly makes their calculations less computationally intensive. At a mathematical level, it also allows us to go further in analyzing their properties, which is an important component of the current paper.

The closest other work we can find to our natural tontine payout is that of Sabin (2010) on a fair annuity tontine. He gives a specific interpretation of ``fair'' and analyzes how to adjust tontine payments for heterogeneous ages and initial contributions. A related fairness question for pooled annuities is treated in Donnelly (2015). Our natural tontine is ``fair'' by construction because we don't mix cohorts or ages. 

In sum, although research on tontine schemes is scattered across the insurance, actuarial, economic, and history journals, we have come across few, if any, scholarly articles that condemn or dismiss the tontine concept outright. \footnote{See recent papers on optimal insurance design, such as Bernard et. al. (2015) and references therein, which focus on traditional risks as opposed to retirement income products. Likewise, Dai, Kwok and Zong (2008) examine the optimal utilization of GLWBs, but not on optimal design.}

\section{Conclusion and Relevance}
\label{conc}

It is not widely known that in the year 1790, the first U.S. Secretary of the Treasury, Alexander Hamilton proposed one of the largest tontines in history. To help reduce a crushing national debt -- something that is clearly not a recent phenomenon -- he suggested the U.S. government replace high-interest revolutionary war debt with new bonds in which coupon payments would be made to a group, as opposed to individuals\footnote{Source: Jennings, Swanson and Trout (1988).}. The group members would share the interest payments evenly amongst themselves, provided they were alive. But, once a member of the group died, his or her portion would stay in a pool and be shared amongst the survivors. This process would theoretically continue until the very last survivor would be entitled to the entire interest payment -- potentially millions of dollars. This obscure episode in U.S. history has become known as Hamilton's Tontine Proposal, which he claimed -- in a letter to George Washington -- would reduce the interest paid on U.S. debt, and eventually eliminate it entirely.

Although Congress decided not to act on Hamilton's proposal \footnote{Hamilton left public life in disgrace after admitting to an affair with a married woman, and soon after died in a duel with Aaron Burr, the U.S. vice president at the time} the tontine idea itself never died on American soil. U.S. insurance companies began issuing tontine insurance policies -- which are close cousins to de Tonti's tontine -- to the public in the mid-19th century, and they became wildly popular\footnote{Source: Ransom and Sutch (1987).}. By the start of the 20th century, historians have documented that half of U.S. households owned a tontine insurance policy, which many used to support themselves through retirement. The longer one lived, the greater their payments. This was a personal hedge against longevity, with little risk exposure for the insurance company. Sadly though, due to shenanigans and malfeasance on the part of company executives, the influential New York State Insurance Commission banned tontine insurance in the state, and by 1910 most other states followed. Tontines have been illegal in the U.S. for over a century and most insurance executives have likely never heard of them.

Tontines not only have a fascinating history but have some basis in economic principles, In fact, Adam Smith himself, writing in the \emph{Wealth of Nations}, noted that tontines may be preferred to life annuities. We believe that a strong case can be made for overturning the current ban on tontine insurance -- \emph{allowing both tontine and annuities to co-exist as they did 320 years ago} -- with suitable adjustments to alleviate problems encountered in the early 20th century. Indeed, given the insurance industry's concern for longevity risk capacity, and its poor experience in managing the risk of long-dated fixed guarantees, one can argue that an (optimal) tontine annuity is a ``triple win" proposition for individuals, corporations and governments. See Milevsky (2015) for further arguments along this line.

It is worth noting that under the proposed (EU) Solvency II guidelines for insurer's capital as well as risk management, there is a renewed focus on {\em total balance sheet} risks. In particular, insurers will be required to hold {\em more} capital against market risk, credit risk and operational risk. In fact, a recently-released report Moody's (2013) claims that \emph{solvency ratios will exhibit a more complex volatility under Solvency II than under Solvency I, as both the available capital and the capital requirements will change with market conditions.} According to many commentators this is likely to translate into higher consumer prices for insurance products with long-term maturities and guarantees. And, although this only applies to European companies (at this point time), it is not unreasonable to conclude that in a global market, annuity loadings will increase, making (participating) tontine products relatively more appealing to price-sensitive consumers.

Moreover, perhaps a properly-designed tontine product could help alleviate the low levels of voluntary annuity purchases -- a.k.a. the annuity puzzle -- by gaming the behavioral biases and prejudices exhibited by retirees. The behavioral economics literature and advocates of cumulative prospect theory have argued that consumers make decisions based on more general \emph{value functions} with personalized decision weights. Among other testable hypotheses, this leads to a preference for investments with (highly) skewed outcomes, even when the alternative is a product with the same expected present values\footnote{Source: Barberis (2013).}. Of course, whether the public and regulators can be convinced of these benefits remains to be seen, but a debate would be informative.

We are not alone in this quest. Indeed, during the last decade a number of companies around the world -- egged on by scholars and journalists\footnote{See for example: Chancellor (2001), Goldsticker (2007), Pechter (2007), Chung and Tett (2007), as well as the more scholarly articles by Richter and Weber (2011), Rotemberg (2009) and especially Sabin (2010).} -- have tried to resuscitate the tontine concept while trying to avoid the bans and taints. Although the specific designs differ from proposal to proposal, all share the same idea we described in this paper: Companies act as custodians and guarantee very little. This arrangement requires less capital which then translates into more affordable pricing for the consumer. Once again the main qualitative insight from our model is that a properly-designed tontine could hold its own in the \emph{utility arena}, against an actuarially-fair life annuity and pose a real challenge to a loaded annuity. The next (natural) step would be to examine the robustness of our results in a full-blown stochastic mortality model. 


\newpage

\section{Appendix and proofs}
\label{apen}

\subsection{Basic Properties}
\label{proof1}

The bulk of this section consists of proofs of results from Section \ref{properties}
To simplify notation, we use as few subscripts as possible, and write $\theta(p)=\theta_{n,\gamma}(p)$, $\beta(p)=\beta_\gamma(p)=\beta_{n,\gamma}(p)$, or $D(p)=D^{\text{OT}}_{n,\gamma}(p)$, as long as the meaning is clear. We start with some calculations for the binomial distribution
\begin{lemma}
\label{derivativeofmeans}
$\frac{d}{dp}E[f(N(p))]=\frac{1}{p}E[g(N(p))]$, \quad\text{where $g(k)=(k-1)(f(k)-f(k-1))$.}
\label{meanderivative}
\end{lemma}
\begin{proof}
\begin{align*}
\frac{d}{dp}E[f(N(p))]
&=\sum_{k=0}^{n-1}\binom{n-1}{k}f(k+1)\Big[kp^{k-1}(1-p)^{n-k-1}-(n-k-1)p^k(1-p)^{n-k-2}\Big]\\
&=\sum_{k=1}^{n-1}\binom{n-1}{k}p^{k-1}(1-p)^{n-k-1}kf(k+1)\\
&\qquad\qquad - \sum_{k=0}^{n-2}\binom{n-1}{k+1}p^{k}(1-p)^{n-k-2}(k+1)f(k+1)\\
&=\frac{1}{p}\sum_{k=1}^{n-1}\binom{n-1}{k}p^k(1-p)^{n-k-1}k[f(k+1)-f(k)].
\end{align*}
\end{proof}

\begin{lemma}
$E[\frac{n}{N(p)}]=\frac{1-(1-p)^n}{p}<\frac{1}{p}$. 
\label{reciprocalbound}
\end{lemma}
\begin{proof}
\begin{align*}
E\Big[\frac{n}{N(p)}\Big]&=\sum_{k=0}^{n-1}\binom{n-1}{k}p^k(1-p)^{n-1-k}\frac{n}{k+1}\\
&=\frac{1}{p}\sum_{k=0}^{n-1}\binom{n}{k+1}p^{k+1}(1-p)^{n-(k+1)}
=\frac{1}{p}[1-(1-p)^n]<\frac{1}{p}.
\end{align*}
\end{proof}

\begin{proof}[Proof of Lemma \ref{betabound}]
Parts (b) and (c) are elementary calculations. To see (a), apply Lemma \ref{derivativeofmeans} to obtain that
\begin{multline*}
\beta'(p)=\frac{d}{dp}E[(N/n)^{\gamma-1}]=\frac{1}{n^{\gamma-1}}\Big(E[N^{\gamma-1}]+p\cdot\frac{1}{p}E[(N-1)(N^{\gamma-1}-(N-1)^{\gamma-1})]\Big)\\
=\frac{1}{n^{\gamma-1}}E[N^\gamma-(N-1)^\gamma]>0.
\end{multline*}
The first statement of (d) holds, because $E[(N/n)^{\gamma-1}]\to P(N=n)=p^{n-1}$. The second part of (d) follows from Lemma \ref{reciprocalbound}. 

Now define 
$$
R_\gamma(p)=\frac{\beta_\gamma(p)^{1/\gamma}}{p},
$$ 
so $R_1(p)=1$. Consider the following:
\begin{lemma} $R_\gamma(p)$ is increasing in $\gamma$. It is increasing in $p$ if $0<\gamma<1$, and is decreasing in $p$ if $\gamma>1$. 
\label{Rlemma}
\end{lemma}
Once this is established, the monotonicity of $R_\gamma$ in $\gamma$ immediately implies (e), completing the proof of Lemma \ref{betabound}. It is worth noting that for $2\le\gamma$ there is a simpler proof of (e), since
$$
\theta(p)=E\Big[\Big(\frac{N(p)}{n}\Big)^{\gamma-1}\Big]\ge E\Big[\frac{N(p)}{n}\Big]^{\gamma-1}=\Big(\frac{1+(n-1)p}{n}\Big)^{\gamma-1}=\Big(p+\frac{1-p}{n}\Big)^{\gamma-1}
>p^{\gamma-1}
$$ 
by Jensen, implying $\beta(p)=p\theta(p)>p^\gamma$. 
\end{proof}

\begin{proof}[Proof of Lemma \ref{Rlemma}]
Let $0<\gamma_2<\gamma_1$. Define $Q$ by $\frac{dQ}{dP}=\frac{pn}{N}$. By Lemma \ref{reciprocalbound}, $Q$ is a sub-probability. Set $b=\frac{\gamma_1}{\gamma_2}>1$. By Holder, 
\begin{multline*}
\beta_{\gamma_2}(p)=pE\Big[\big(\frac{N(p)}{n}\big)^{\gamma_2-1}\Big]
=E_Q\Big[\big(\frac{N(p)}{n}\big)^{\gamma_2}\Big]
\le E_Q\Big[\big(\frac{N(p)}{n}\big)^{\gamma_2 b}\Big]^{\frac{1}{b}}E_Q[1]^{1-\frac{1}{b}}\\
< E_Q\Big[\big(\frac{N(p)}{n}\big)^{\gamma_1}\Big]^{\gamma_2/\gamma_1}
=\Big(pE\Big[\big(\frac{N(p)}{n}\big)^{\gamma_1-1}\Big]\Big)^{\gamma_2/\gamma_1}
=\beta_{\gamma_1}(p)^{\gamma_2/\gamma_1}.
\end{multline*}
Taking a $1/\gamma_2$ power and dividing by $p$ now shows that $R_{\gamma_2}(p)\le R_{\gamma_1}(p)$. 

To show the monotonicity properties of $R_\gamma$ with respect to $p$, it suffices to show the same for $\log R_\gamma$. As in the proof of (a) of Lemma \ref{betabound}, 
\begin{multline*}
\frac{d}{dp} \log R_\gamma(p) = \frac{\beta'(p)}{\gamma\beta(p)}-\frac{1}{p}
=\frac{1}{\gamma \beta(p)}\Big[\beta'(p)-\frac{\gamma \beta(p)}{p}\Big]\\
=\frac{1}{\gamma\beta(p)n^{\gamma-1}}E\Big[N(p)^\gamma-(N(p)-1)^\gamma-\gamma N(p)^{\gamma-1}\Big].
\end{multline*}
By Taylor's theorem, $k^\gamma-(k-1)^\gamma-\gamma k^{\gamma-1}=-\frac{\gamma(\gamma-1)}{2}\xi^{\gamma-2}$ for some $\xi$. This is $<0$ when $\gamma>1$ and $>0$ if $0<\gamma<1$, which completes the proof.  
\end{proof}

Turning to conjecture \eqref{depletioninequality}, define
$$
R_{\gamma_1,\gamma_2}(p)=\frac{\beta_{\gamma_1}(p)^{1/\gamma_1}}{\beta_{\gamma_2}(p)^{1/\gamma_2}},
$$
so $R_\gamma(p)=R_{\gamma,1}(p)$ and $R_{\gamma_1,\gamma_2}(1)=1$. For $Q$ as in the proof of Lemma \ref{Rlemma}, an application of Holder's inequality (as in that proof) shows that 
\begin{equation}
R_{\gamma_1,\gamma_2}(p)=E_Q[(N(p)/n)^{\gamma_1}]^{1/\gamma_1}/E_Q[(N(p)/n)^{\gamma_2}]^{1/\gamma_2}\ge 1,
\label{Rbiggerthanone}
\end{equation} 
when $\gamma_1>\gamma_2$.
\begin{lemma}
Assume $R_{\gamma_1,\gamma_2}(p)$ decreases with $p$, for some given $n$ and $\gamma_1>\gamma_2$. 
Then \eqref{depletioninequality} also holds. In other words, $\Delta_{n,\gamma_1}(t)<\Delta_{n,\gamma_2}(t)$ for every $t>0$. 
\label{Rmonotonicity}
\end{lemma}
\begin{proof}
There is a constant $C$ [depending on $n$, $\gamma_1$, $\gamma_2$, and the mortality parameters] such that $D_{\gamma_1}(p)=C D_{\gamma_2}(p)R_{\gamma_1,\gamma_2}(p)$ for each $p$. Therefore 
$$
\Delta_{n,\gamma_1}(t)-\Delta_{n,\gamma_2}(t)=\int_0^t e^{-rs}D_{\gamma_2}({}_sp_x)\Big[CR_{\gamma_1,\gamma_2}({}_sp_x)-1\Big]\,ds.
$$
By \eqref{tontineconstraint}, this $\to 1-1=0$ in the limit as $t\to\infty$, so the integrand must take both positive and negative values. By \eqref{Rbiggerthanone} we have $R_{\gamma_1,\gamma_2)}({}_sp_x)\ge 1$ for each $s$, so $0<C<1$. Therefore $CR_{\gamma_1,\gamma_2}({}_sp_x)-1$ starts out negative, and then increases with $s$, eventually becoming positive.  By hypothesis, once it becomes positive, it stays so. There is thus a $t_0$ such that the integrand is negative for $s<t_0$ and positive for $s>t_0$. In other words, $\Delta_{n,\gamma_1}(t)-\Delta_{n,\gamma_2}(t)$ is decreasing on $[0,t_0]$ and increasing on $[t_0,\infty)$. Since $\Delta_{n,\gamma_1}(0)-\Delta_{n,\gamma_2}(0)=0$ and $\Delta_{n,\gamma_1}(t)-\Delta_{n,\gamma_2}(t)\to0$ when $t\to\infty$, we see that $\Delta_{n,\gamma_1}(t)-\Delta_{n,\gamma_2}(t)<0$ for every $t>0$, as required.
\end{proof}

\begin{proof}[Proof of Theorem \ref{increasinggamma}]
For part (a), let $\gamma_1\ge 1\ge \gamma_2>0$, with $\gamma_1\neq \gamma_2$. Then $R_{\gamma_1,\gamma_2}(p)=\frac{R_{\gamma_1}(p)}{R_{\gamma_2}(p)}$ decreases with $p$, by the second statement of Lemma \ref{Rlemma}. Now apply Lemma \ref{Rmonotonicity}

To show part (b) it suffices to show that $R_{\infty,\gamma_2}(p)$ is decreasing in $p$. This follows from (a) and (d) of Lemma \ref{betabound}, since $R_{\infty,\gamma_2}(p)=1/\beta_{\gamma_2}(p)^{1/\gamma_2}$. 
Part (c) is implicit in the proof of Lemma \ref{Rmonotonicity}, since the hypothesis of that result is not used in showing that $\Delta_{n,\gamma_1}(t)-\Delta_{n,\gamma_2}(t)$ initially decreases.
\end{proof}

Numerical evidence (not reported here) suggests that there is a $\gamma_0$ ($\approx 7$) such that $R_{\gamma_1,\gamma_2}(p)$ is decreasing in $p$ whenever $\gamma_0>\gamma_1>\gamma_2$. By Lemma \ref{Rmonotonicity}, this would imply  
\eqref{depletioninequality} for these parameters, without any restriction on the underlying mortality distribution. However, for larger $\gamma$ values, $R_{\gamma_1,\gamma_2}(p)$ can in fact increase over a small range of $p$'s. For example, with $n=2$, $\gamma_1=10$ and $\gamma_2=9$, it increases over the range $p\in[0.00098, 0.00394]$. 
This defect in $R$ for very small $p$ does not appear significant enough to cause \eqref{depletioninequality} to fail under any reasonable mortality distribution. But it is in fact possible to construct a pathological mortality distribution for which \eqref{depletioninequality} fails, by making ${}_tp_x$ spend essentially all of its time in the above range of $p$'s. However, this distribution's hazard rate will not be monotone, since we achieve it by taking the hazard rate low over the above range of $p$'s, and high elsewhere. 

Asymptotics in $n$ are problematic, because of the need to estimate $\beta(p)$ for $p$ very small. We circumvent this difficulty by capping tontine payments at some advanced age $x+T$. For a tontine running $T$ years, the optimal payout $\tilde d(t)$ will maximize $\int_0^T e^{-rs}{}_sp_x \frac{\tilde d(s)^{1-\gamma}}{1-\gamma}\theta_\gamma({}_sp_x)\,ds$ over $\tilde d$ satisfying the constraint $\int_0^Te^{-rs}\tilde d(s)\,ds=1$. The arguments of Corollary \ref{CRRAoptimum} now show that 
$$
\tilde d(t)=\frac{\beta_\gamma({}_tp_x)^{1/\gamma}}{\int_0^Te^{-rs}\beta_\gamma({}_sp_x)^{1/\gamma}\,ds}.
$$
We now set $\tilde \Delta_{n,\gamma}(t)=\int_0^t e^{-rs}\tilde d(s)\,ds$, and can prove the following asymptotic version of \eqref{depletioninequality}.
\begin{theorem}
\label{asymptotictheorem}
Let $\gamma_1>\gamma_2$. For every $\epsilon>0$ there is an $n_0$ such that $\tilde \Delta_{n,\gamma_1}(t)<\tilde \Delta_{n,\gamma_2}(t)$ for each $t\in[\epsilon,T-\epsilon]$ and $n\ge n_0$.
\end{theorem}
\begin{proof} We start by establishing a number of estimates. Fix $\gamma>0$, $1>\eta>0$, $1>p_1>0$, and $C_1>0$. We may then find $n_1$ such that if $n\ge n_1$ and $1\ge p\ge p_1$ then $2(1\lor n^{1-\gamma})e^{-\frac12(n-1)\eta^2p^2}\le\frac{C_1}{n^{3/2}}$.
Recall Azuma's inequality (see for example Steele (1997)), in the form that if $B\sim \text{Bin}(k,p)$ then $P(|\frac{B}{k}-p|>b)\le 2e^{-\frac12kb^2}$. Observe that $\frac1n\le \frac{N(p)}{n}\le 1$, and set $k=n-1$, $B=N(p)-1$, and $\mu=p+\frac{1-p}{n}$. It follows that if $n\ge n_1$ and $1\ge p\ge p_1$ then 
\begin{multline}
E\Big[\big|\frac{N(p)}{n}\big|^{\gamma-1}, \big|\frac{N(p)}{n}-\mu\big|>\eta p\Big]
\le(1\lor n^{1-\gamma})P\Big(\big|\frac{B-kp}{n}\big|>\eta p\Big)\\
\le(1\lor n^{1-\gamma})P\Big(\big|\frac{B}{k}-p\big|>\eta p\Big)
\le \frac{C_1}{n^{3/2}}.
\label{firstbound}
\end{multline}
Likewise, if $n\ge n_1$ and $1\ge p\ge p_1$ then 
$E[|\frac{N(p)}{n}-\mu|^a, |\frac{N(p)}{n}-\mu|>\eta p]\le \frac{C_1}{n^{3/2}}$ for each $a\ge 0$, so 
\begin{equation}
E\Big[\Big(\frac{N(p)}{n}-\mu\Big)^a, \big|\frac{N(p)}{n}-\mu\big|\le \eta p\Big]=
\begin{cases} 
1+e_0, & a=0\\
e_1, &a=1\\
\frac{kp(1-p)}{n^2}+e_2, &a=2
\end{cases}
\label{secondbound}
\end{equation}
where in each case $|e_i|\le \frac{C_1}{n^{3/2}}$.

By Azuma's inequality again, 
\begin{align}
E\Big[\big|\frac{N(p)}{n}-\mu\big|^3\Big]
&= \int_0^\infty P\Big(\big|\frac{N(p)}{n}-\mu\big|^3>q\Big)\,dq
= \int_0^\infty P\Big(\big|\frac{B-kp}{n}\big|>q^{1/3}\Big)\,dq\nonumber\\
&\le \int_0^\infty P\Big(\big|\frac{B}{k}-p\big|>q^{1/3}\Big)\,dq
\le 2\int_0^\infty e^{-\frac12kq^{2/3}}\,dq \label{thirdbound}\\
&=\frac{6}{k^{3/2}}\int_0^\infty e^{-\frac12 z^2}z^2\,dz
=\frac{3\sqrt{2\pi}}{k^{3/2}}.\nonumber
\end{align}
By Taylor's theorem, 
\begin{multline*}
E\Big[\Big(\frac{N}{n}\Big)^{\gamma-1}, \big|\frac{N}{n}-\mu\big|\le\eta p\Big]=E\Big[\mu^{\gamma-1}+(\gamma-1)\mu^{\gamma-2}\Big(\frac{N}{n}-\mu\Big)
+\frac{(\gamma-1)(\gamma-2)}{2}\mu^{\gamma-3}\Big(\frac{N}{n}-\mu\Big)^2\\
+\frac{(\gamma-1)(\gamma-2)(\gamma-3)}{6}\xi^{\gamma-4}\Big(\frac{N}{n}-\mu\Big)^3, \big|\frac{N}{n}-\mu\big|\le\eta p\Big]
\end{multline*}
where $\xi$ lies between $(1-\eta)\mu$ and $(1+\eta)\mu$. Applying \eqref{firstbound}, \eqref{secondbound}, and \eqref{thirdbound} yields a constant $C_2$ such that 
if $n\ge n_1$ and $1\ge p\ge p_1$ then 
$$
E\Big[\Big(\frac{N(p)}{n}\Big)^{\gamma-1}\Big]
=\mu^{\gamma-1}+\frac{(\gamma-1)(\gamma-2)}{2}\mu^{\gamma-3}\frac{kp(1-p)}{n^2}+e_4,
$$
where $|e_4|\le \frac{C_2}{n^{3/2}}$.
Since $\mu=p+\frac{1-p}{n}$ it is now immediate that for these $n$'s and $p$'s, also
$$
E\Big[\Big(\frac{N(p)}{n}\Big)^{\gamma-1}\Big]
=p^{\gamma-1}+\frac{\gamma(\gamma-1)(1-p)}{2n}p^{\gamma-2}+e_5=p^{\gamma-1}\Big(1+\frac{\gamma(\gamma-1)(1-p)}{2np}+\frac{e_5}{p^{\gamma-1}}\Big)
$$
where $|e_5|\le \frac{C_3}{n^{3/2}}$ for some constant $C_3$. By passing to some $n_2\ge n_1$ if necessary, we obtain finally that
\begin{multline}
\beta_\gamma(p)^{1/\gamma}=\Big(pE\Big[\big(\frac{N(p)}{n}\big)^{\gamma-1}\Big]\Big)^{1/\gamma}\\
=p\Big(1+\frac{(\gamma-1)(1-p)}{2np}+e_6\Big)
=p\Big(1-\frac{(\gamma-1)}{2n}\Big)+\frac{\gamma-1}{2n}+pe_6
\label{fourthbound}
\end{multline}
for $n\ge n_2$ and $1\ge p\ge p_1$, where $|e_6|\le\frac{C_4}{n^{3/2}}$, and $C_4$ is a constant depending on $\gamma$ and $p_1$. 

To apply this, fix $\gamma_1>\gamma_2$, set $p_1={}_{T}p_x$, and choose $n_2$ and $C_4$ to ensure that \eqref{fourthbound} holds for both $\gamma_1$ and $\gamma_2$. 
For any $0<t<T$, the inequality $\tilde \Delta_{n,\gamma_1}(t)<\tilde \Delta_{n,\gamma_2}(t)$ is equivalent to the following:
\begin{multline}
\int_0^t e^{-rs}\beta_{\gamma_1}({}_sp_x)^{1/\gamma_1}\,ds\times \int_0^T e^{-rq}\beta_{\gamma_2}({}_qp_x)^{1/\gamma_2}\,dq\\
< \int_0^t e^{-rs}\beta_{\gamma_2}({}_sp_x)^{1/\gamma_2}\,ds\times \int_0^T e^{-rq}\beta_{\gamma_1}({}_qp_x)^{1/\gamma_1}\,dq.
\label{tinequality}
\end{multline}
Set $F(t)=\int_0^t e^{-rs}{}_sp_x\,ds$ and $G(t)=\int_0^t e^{-rs}\,ds$. Using \eqref{fourthbound}, we have an upper bound of 
$$
\Big[F(t)(1-\frac{(\gamma_1-1)}{2n}\Big)+G(t)\frac{\gamma_1-1}{2n}\Big]\Big[F(T)(1-\frac{(\gamma_2-1)}{2n}\Big)+G(T)\frac{\gamma_2-1}{2n}\Big]+\frac{C_5}{n^{3/2}}
$$
for the left side of \eqref{tinequality} (for some constant $C_5$), provided $n\ge n_2$. There is likewise a lower bound
$$
\Big[F(t)(1-\frac{(\gamma_2-1)}{2n}\Big)+G(t)\frac{\gamma_2-1}{2n}\Big]\Big[F(T)(1-\frac{(\gamma_1-1)}{2n}\Big)+G(T)\frac{\gamma_1-1}{2n}\Big]-\frac{C_6}{n^{3/2}}
$$
for the right side (for some constant $C_6$), again provided $n\ge n_2$. Thus \eqref{tinequality} will hold for $t$ and $n$, provided $n\ge n_2$ and 
$$
\frac{C_5+C_6}{n^{3/2}}<
\frac{\gamma_1-\gamma_2}{2n}\Big[F(t)G(T)-G(t)F(T)\Big].
$$
It is easily seen that $F(t)G(T)-G(t)F(T)=\int_0^t\int_t^Te^{-r(s+q)}({}_sp_x-{}_qp_x)\,dq\,ds$, and it follows that $F(t)G(T)-G(t)F(T)>0$ on $(0,T)$. It is therefore clear how to choose $n_0\ge n_2$ to ensure that the desired conclusion holds. 
\end{proof}

We close this section with the proof of the result in Section \ref{tontineannuitycagematch}
\begin{proof}[Proof of Theorem \ref{utilityinequality}]
As remarked earlier, the case $\gamma\neq 1$ follows immediately from (e) of Lemma \ref{betabound}. When $\gamma=1$, we have
$u(c)=\log c$ so $U_{n,1}^{\text{OT}}=\int_0^\infty e^{-rt}{}_tp_x \,E[\log(\frac{nc_0\,{}_tp_x}{N({}_tp_x)})]\,dt$, while 
$U_1^{\text{A}}=\int_0^\infty e^{-rt}{}_tp_x \,\log(c_0)\,dt$. Therefore 
$$
U_1^{\text{A}}-U_{n,1}^{\text{OT}}=\int_0^\infty e^{-rt}{}_tp_x \Big(E\Big[\log\big(\frac{N({}_tp_x)}{n}\big)\Big]-\log({}_tp_x)\Big)\,dt.
$$
The following result now completes the proof. 
\end{proof}

\begin{lemma}
$E[\log(N(p))]>\log(np)$ for $0<p<1$.
\end{lemma}
\begin{proof}
By Holder, $E\big[(\frac{n}{N(p)})^a\big]^{1/a}$ increases with $a$. By Lemma \ref{reciprocalbound} and l'H\^opital's rule,
\begin{align*}
\frac{1}{p}&>\lim_{a\downarrow 0}E[(\frac{n}{N(p)})^a]^{1/a}
=e^{\lim_{a\downarrow 0}\frac{1}{a}\log E[(\frac{n}{N(p)})^a]}\\
&=e^{\lim_{a\downarrow 0}E[(\frac{n}{N(p)})^a\log(\frac{n}{N(p)})]/E[(\frac{n}{N(p)})^a]}
=e^{-E[\log(\frac{N(p)}{n})]}.
\end{align*}
Taking log's shows the result.
\end{proof}

\subsection{Indifference Loading}
\begin{proof}[Proof of Theorem \ref{loadinginequality}]
Since $1<\gamma\le 2$, $c^{\gamma-1}$ is concave in $c$, so lies below its tangents. Therefore
$$
\theta(p)=E\Big[\Big(\frac{N(p)}{n}\Big)^{\gamma-1}\Big]
\le E\Big[p^{\gamma-1}+(\gamma-1)p^{\gamma-2}\Big(\frac{N}{n}-p\Big)\Big]
=p^{\gamma-1}+(\gamma-1)p^{\gamma-2}\frac{1-p}{n}.
$$
$c^{1/\gamma}$ is also strictly concave in $c$, so in the same way
$$
\beta(p)^{\frac{1}{\gamma}}=(p\theta(p))^{\frac{1}{\gamma}}
<(p^\gamma)^{\frac{1}{\gamma}}+\frac{1}{\gamma}(p^\gamma)^{\frac{1}{\gamma}-1}\cdot (\gamma-1)p^{\gamma-1}\frac{1-p}{n}
=p+\frac{(\gamma-1)(1-p)}{\gamma n}.
$$
Therefore 
$$
\int_0^\infty e^{-rt}\beta({}_tp_x)^{\frac{1}{\gamma}}\,dt < \frac{1}{c_0} + \frac{\gamma-1}{\gamma n}\Big(\frac{1}{r}-\frac{1}{c_0}\Big).
$$
By definition, 
$$
\frac{c_0^{-\gamma}}{1-\gamma}(1-\delta)^{1-\gamma}
=\frac{1}{1-\gamma}\Big(\int_0^\infty e^{-rt}\beta({}_tp_x)^{1/\gamma}\,dt\Big)^\gamma,
$$
and $1-c^{\frac{\gamma}{1-\gamma}}$ is also concave in $c$, so as before 
$$
\delta=1-\Big(c_0\int_0^\infty e^{-rt}\beta({}_tp_x)^{\frac{1}{\gamma}}\,dt\Big)^{\frac{\gamma}{1-\gamma}}\cdot 
<-\frac{\gamma}{1-\gamma}\cdot \frac{\gamma-1}{\gamma n}\Big(\frac{c_0}{r}-1\Big)=\frac{1}{n}\Big(\frac{c_0}{r}-1\Big).
$$
 \end{proof}
For any $\gamma$ one can derive (using the second moment of $N(p)$ now, as in the proof of Theorem \ref{asymptotictheorem}) the asymptotic result that $\delta n\to \frac{\gamma}{2}(\frac{c_o}{r}-1)$. But convergence turns out to be so slow that this precise asymptotic is of limited use. The slow convergence derives from the observation that, with Gompertz mortality, the time $t$ till ${}_tp_x$ reaches $\frac{1}{n}$ grows only at rate $b\log\log n$. (For the same reason, numerical computations are best carried out using $\log({}_tp_x)$ rather than ${}_tp_x$.) For example, with $\gamma=2$, $r=3\%$, age $x=50$, and Gompertz parameters $m=87.25$ and $b=9.5$, we obtain $\frac{\gamma}{2}(\frac{c_o}{r}-1)=0.6593$; But for $n=$ 10, 100, or 1000 we only have $\delta n=$ 0.2858, 0.3377, and 
0.3671; In fact, even with $n=7\times 10^9$ tontine participants (roughly the world's entire population)
we would only reach $\delta n = 0.4417$. 

Asymptotics become more reliable if tontine and annuity payments are capped at some advanced age $x+T$, as in Theorem \ref{asymptotictheorem}. The new loading $\tilde \delta$ now satisfies $\tilde \delta n\to \frac{\gamma}{2}(\frac{\tilde c_o}{r}(1-e^{-rT})-1)$, where $\tilde c_0$ is the payout from the capped annuity. With the same parameters as above and capping at age 120 it still takes $n=100,\!000$ to give $n\tilde\delta=0.4012$ (well approximating the $n=\infty$ value of $0.4301$). Convergence is more rapid with smaller $T$'s: with capping at age 110 we only need $n=1,\!000$ to give $n\tilde\delta=0.3642$ (approximating 0.3850) and for age 100 capping, even $n=100$ yields 0.2855 (approximating 0.2897).

\subsection{Certainty Equivalents and the Natural Tontine}
\label{moreoncertaintyequiv}
In Table  \ref{table08} we compare the welfare loss experienced by an individual with Longevity Risk Aversion $\gamma\neq 1$, if they participate in a natural tontine rather than an optimal one. To do so we calculate the ratio $\Gamma\ge 1$ of the certainty equivalents of the two tontines. This represents the initial deposit into a natural tontine needed to provide the same utility as a \$1 deposit into an optimal one. In this section, we discuss this comparison further. 

The natural tontine has utility 
$$
U_{n,\gamma}^{\text{N}}=\frac{c_0^{1-\gamma}}{1-\gamma}\int_0^\infty e^{-rt}{}_tp_x^{2-\gamma}\theta_{n,\gamma}({}_tp_x)\,dt.
$$
Therefore
\begin{align*}
\Gamma
=\left[\frac{U^{\text{OT}}_{n,\gamma}}{U^{\text{N}}_{n,\gamma}}\right]^{\frac{1}{1-\gamma}}
&=\left[\frac{D^{\text{OT}}_{n,\gamma}(1)^{-\gamma}}{c_0^{1-\gamma}\int_0^\infty e^{-rt}{}_tp_x^{2-\gamma}\theta_{n,\gamma}({}_tp_x)\,dt}\right]^{\frac{1}{1-\gamma}}\\
&=\frac{\big(\int_0^\infty e^{-rt}{}_tp_x\,dt\big)\big(\int_0^\infty e^{-rt}\beta_{n,\gamma}({}_tp_x)^{\frac{1}{\gamma}}\,dt)^{\frac{\gamma}{1-\gamma}}}{\big(\int_0^\infty e^{-rt}{}_tp_x^{2-\gamma}\theta_{n,\gamma}({}_tp_x)\,dt\big)^{\frac{1}{1-\gamma}}}.
\end{align*}

When we compute these values numerically, for $0<\gamma\le 2$, we get values quite close to 1. Moreover, if we let 
$n\to\infty$, then 
$\beta_{n,\gamma}(p)^{1/\gamma}\to p$ and $\theta_{n,\gamma}(p)\to p^{-(1-\gamma)}$, which makes $\Gamma\to 1$ asymptotically, as long as $\gamma\le 2$.

For $\gamma>2$ this is not a fair comparison, as the integral diverges and $U_{n,\gamma}^{\text{N}}=\infty$. It is worth understanding why. The integrand in the numerator involves $[p\theta_{n,\gamma}(p)]^{1/\gamma}\sim p^{\frac{1}{\gamma}}n^{\frac{1-\gamma}{\gamma}}$ which is well behaved as $p\to 0$. The denominator integrand involves $p^{2-\gamma}\theta_{n,\gamma}(p)$ however, which can be large while $p\to 0$, if $\gamma>2$. This means that for $\gamma>2$ the natural tontine utility is unduly influenced by the possibility of living to highly advanced ages. Even if there is only a single survivor, that survivor's payout will have dropped to quite low levels by the time it is actually improbable that anyone will live that long. And for $\gamma>2$ the negative consequences of the low payout dominate the small probability of surviving that long. 

Capping tontine payouts at an advanced age $x+T$, as in Theorem \ref{asymptotictheorem}, would enable comparisons for $\gamma>2$. For example, take age $x=60$ and let us compare with the certainty equivalent $\Gamma=1.0034$ at that age from Table \ref{table08} (ie. with $n=100$ and $\gamma=2$). Capping payouts at age 100, we can achieve $\tilde\Gamma=1.0032$ for $\gamma=4$ even with $n=50$, and $\tilde\Gamma=1.0037$ for $\gamma=10$ by going to $n=300$. With the capping age raised to 110 and $\gamma=4$, it instead takes $n=1400$ to achieve a comparable $\tilde\Gamma=1.0032$

\newpage

\begin{table}
\begin{center}
\begin{tabular}{||c||c||c||c||}
\hline\hline
\multicolumn{4}{||c||}{\textbf{Optimal Tontine Payout Function: Pool of Size}
$n=25$} \\ \hline\hline
LoRA ($\gamma $) & Payout Age 65 & Payout Age 80 & Payout
Age 95 \\ \hline\hline
0.5 & 7.565\% & 5.446\% & 1.200\% \\ \hline\hline
1.0 & 7.520\% & 5.435\% & 1.268\% \\ \hline\hline
1.5 & 7.482\% & 5.428\% & 1.324\% \\ \hline\hline
2.0 & 7.447\% & 5.423\% & 1.374\% \\ \hline\hline
4.0 & 7.324\% & 5.410\% & 1.541\% \\ \hline\hline
9.0 & 7.081\% & 5.394\% & 1.847\% \\ \hline\hline
\textbf{Survival} & ${}_0p_{65}=$100\% & ${}_{15}p_{65}=$72.2\% & ${}_{30}p_{65}=$16.8\% \\ \hline\hline
\multicolumn{4}{||c||}{\footnotesize Notes: Assumes $r=4\%$ and Gompertz Mortality ($m=88.72,b=10$)} \\ \hline\hline
\end{tabular}\medskip
\caption{Shows the optimal tontine payout function $d(t)=D_{n,\gamma}^{\text{OT}}({}_tp_x)$ at various ages. Note how insensitive it is to Longevity Risk Aversion $\gamma$, even for a small pool ($n=25$).}
\label{table06}
\end{center}
\end{table}

.
\newpage
.

\begin{table}[h]
\begin{center}
\begin{tabular}{||c||c|c|c|c|c||}
\hline\hline
\multicolumn{6}{||c||}{{\bf The Highest Annuity Loading $\delta$ You Are Willing to Pay}} \\ \hline\hline
\multicolumn{6}{||c||}{{\bf If a Tontine Pool of Size $n$ is Available}} \\ \hline\hline
LoRA $\gamma$ & $n=20$ & $n=100$ & $n=500$ & $n=1000$ & $n=5000$\\ \hline\hline
0.5 & 72.6 b.p. & 14.5 b.p. & 2.97 b.p. & 1.50 b.p. & 0.30 b.p. \\ \hline
1.0 &   129.8 b.p. & 27.4 b.p. & 5.74 b.p. & 2.92 b.p. & 0.60 b.p.  \\ \hline
1.5 & 182.4 b.p. & 39.8 b.p. & 8.45 b.p. & 4.31 b.p. & 0.89 b.p. \\ \hline
2.0 & 231.7 b.p. & 51.8 b.p. & 11.1 b.p. & 5.68 b.p. & 1.18 b.p. \\ \hline
3.0 & 323.1 b.p. &  75.1 b.p. &  16.3 b.p. & 8.38 b.p. & 1.75 b.p. \\ \hline
9.0 & 753.6 b.p. &  199.8 b.p. &  45.9 b.p. &  23.8 b.p. &  5.09 b.p. \\ \hline
\hline
\multicolumn{6}{||c||}{\footnotesize Assumes Age $x=60$, $r=3\%$ and Gompertz Mortality ($m=87.25,b=9.5$)} \\ \hline\hline
\end{tabular}\medskip
\caption{If the risk loading $\delta$ the annuity charges up front is higher than these amounts, the tontine is preferred.}
\label{table07}
\smallskip
\end{center}
\end{table}

.
\newpage
.

\begin{table}[h]
\begin{center}
\begin{tabular}{||c||c|c|c||}
\hline\hline
\multicolumn{4}{||c||}{{\bf Natural vs. Optimal Tontine}} \\ \hline\hline
\multicolumn{4}{||c||}{{\bf Certainty Equivalent for $n=100$}} \\ \hline\hline
Age $x$ & $\gamma=0.5$ & $\gamma=1$ & $\gamma=2$  \\ \hline\hline
30  &  1.000018	 &  1	 &  1.000215 \\  \hline
40  &  1.000026	 &  1	 &  1.000753 \\  \hline
50  &  1.000041	 &  1	 &  1.001674 \\  \hline
60  &  1.000067	 &  1	 &  1.003388 \\  \hline
70  &  1.000118	 &  1	 &  1.003451 \\  \hline
80  &  1.000225	 &  1	 &  1.009877\\
\hline\hline
\multicolumn{4}{||c||}{\footnotesize $r=3\%$ and Gompertz $m=87.25,b=9.5$} \\ \hline\hline
\end{tabular}\medskip
\caption{This shows how much must be invested in a natural tontine to match \$1 invested in a tontine optimized for the individual's LoRA $\gamma$.} 
\label{table08}
\smallskip
\end{center}
\end{table}

\begin{figure}[here]
\begin{center}
\includegraphics[width=0.9\textwidth]{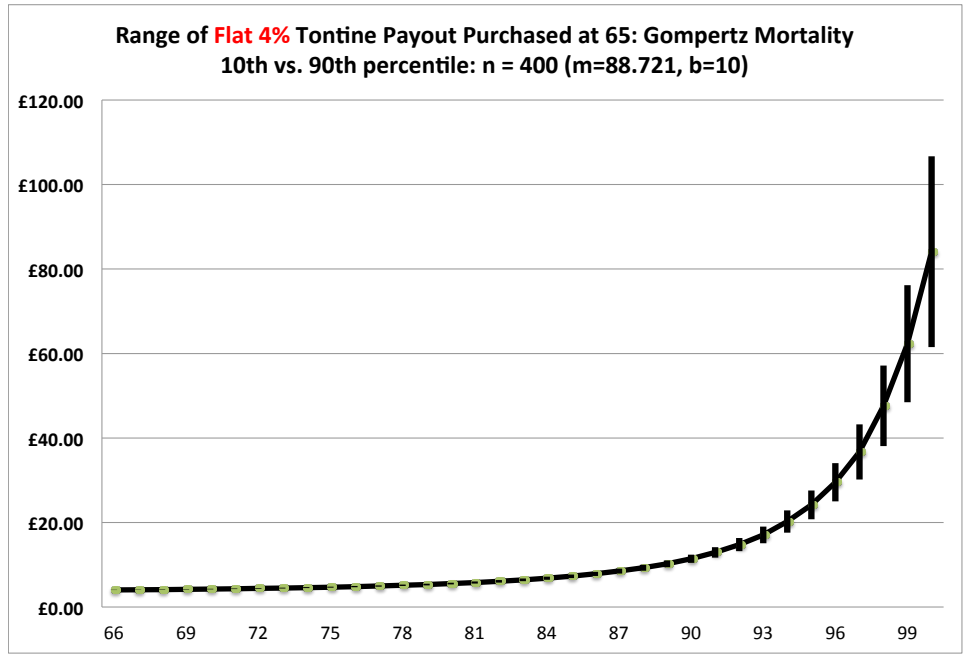} 
\caption{The range of $d_0=4\%$ tontine dividends during the first few decades of retirement is relatively low and predictable for a pool size in the hundreds. The dividends increase exponentially at later ages and the 80\% range is much wider as well. But this is not the only way to construct a tontine.}
\label{fig1}
\end{center}
\end{figure}

\begin{figure}[here]
\begin{center}
\includegraphics[width=0.9\textwidth]{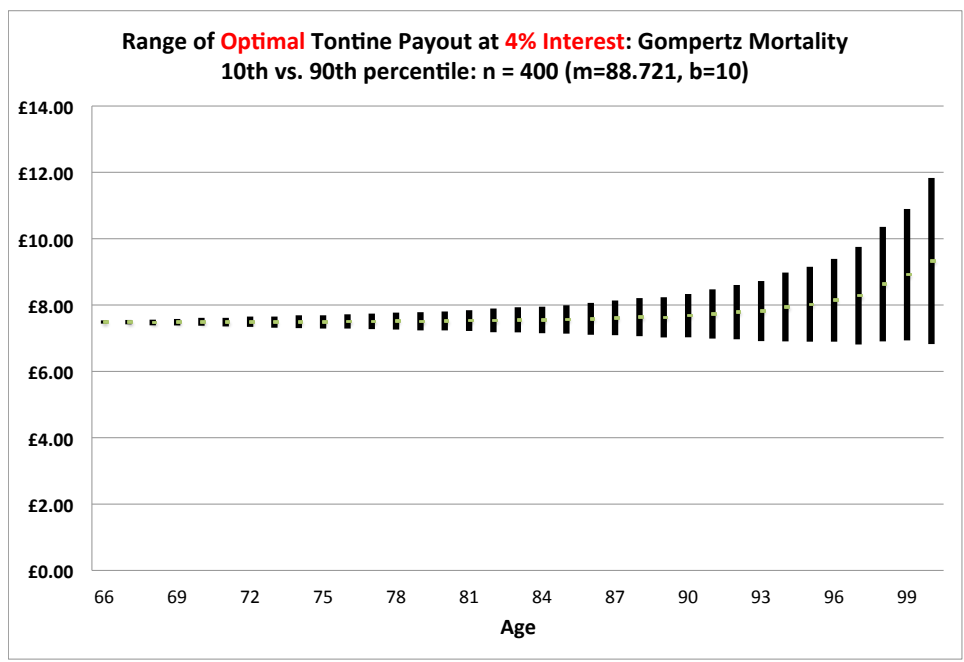} 
\caption{The optimal tontine pays survivors a cash value that is expected to remain relatively constant over time, conditional on survival; although the 80\% range of outcomes does increase at higher ages given the inherent uncertainty in the number of survivors from an initial pool of $n=400$. This structure is optimal for logarithfig3mic $\gamma=1$ utility and nearly optimal for all other levels of Longevity Risk Aversion.}
\label{fig2}
\end{center}
\end{figure}

\begin{figure}[here]
\begin{center}
\includegraphics[width=0.9\textwidth]{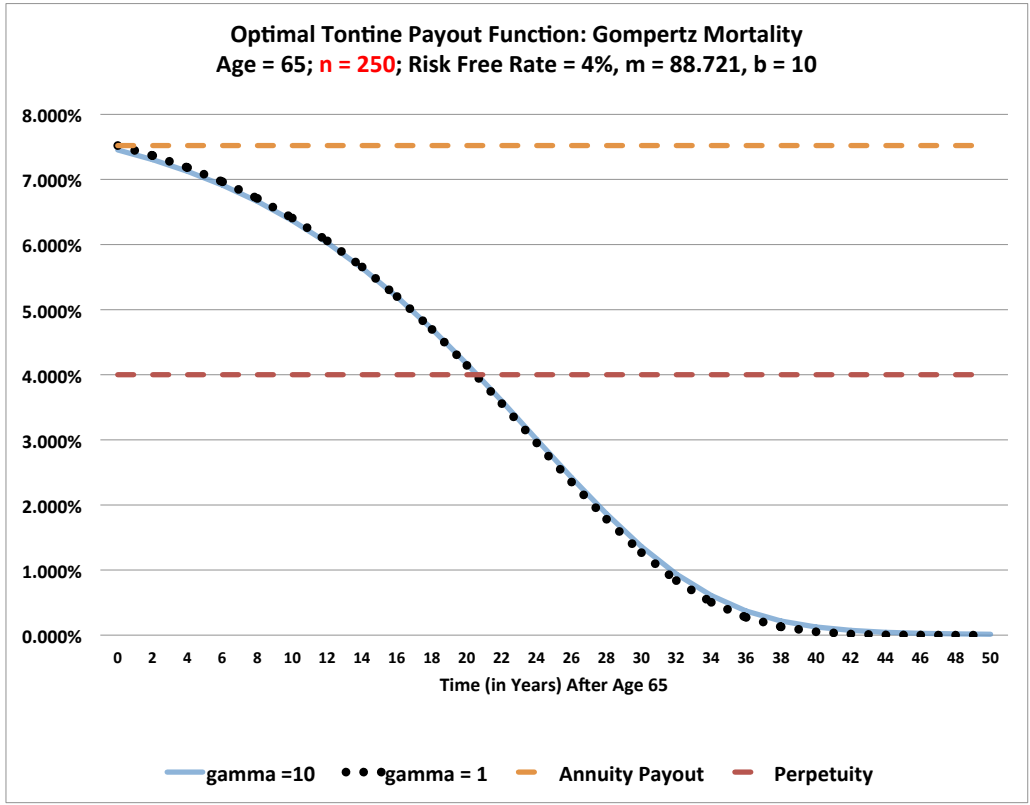} 
\caption{The difference between the optimal tontine payout function $d(t)$ for different levels of Longevity Risk Aversion $\gamma$ is barely noticeable once the tontine pool size is greater than $n=250$.}
\label{fig3}
\end{center}
\end{figure}

\begin{figure}[here]
\begin{center}
\includegraphics[width=0.9\textwidth]{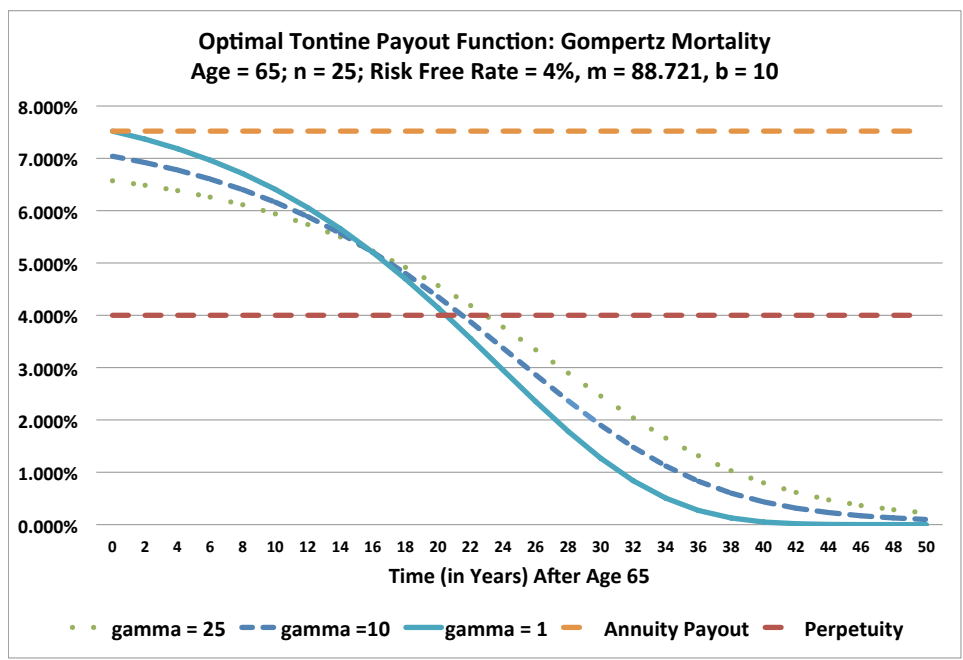} 
\caption{The optimal tontine payout function for someone with logarithmic $\gamma=1$ utility starts off paying the exact same rate as a life annuity regardless of the number of participants in the tontine pool. But, for higher levels of longevity risk aversion $\gamma$ and a relatively smaller tontine pool, the function starts off at a lower value and declines at a slower rate}.
\label{fig4}
\end{center}
\end{figure}

\begin{figure}[here]
\begin{center}
\includegraphics[width=0.9\textwidth]{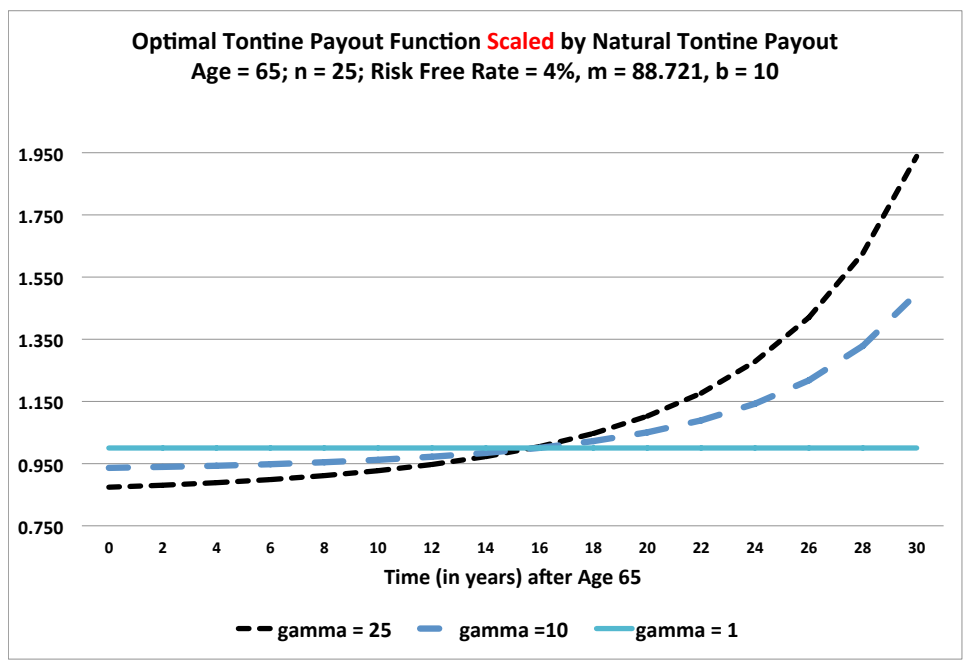} 
\caption{The ratio $D_{n,\gamma}^{\text{OT}}(_tp_x)/D_{\text{N}}({}_tp_x)$ of the optimal tontine payout to the natural tontine payout. This is flat for $\gamma=1$ since the natural tontine is optimal for logarithmic utility, but for rising $\gamma$ we see increased reserving against the risk of living longer than expected.}
\label{fig5}
\end{center}
\end{figure}

\newpage

\begin{figure}[here]
\begin{center}
\includegraphics[width=0.9\textwidth]{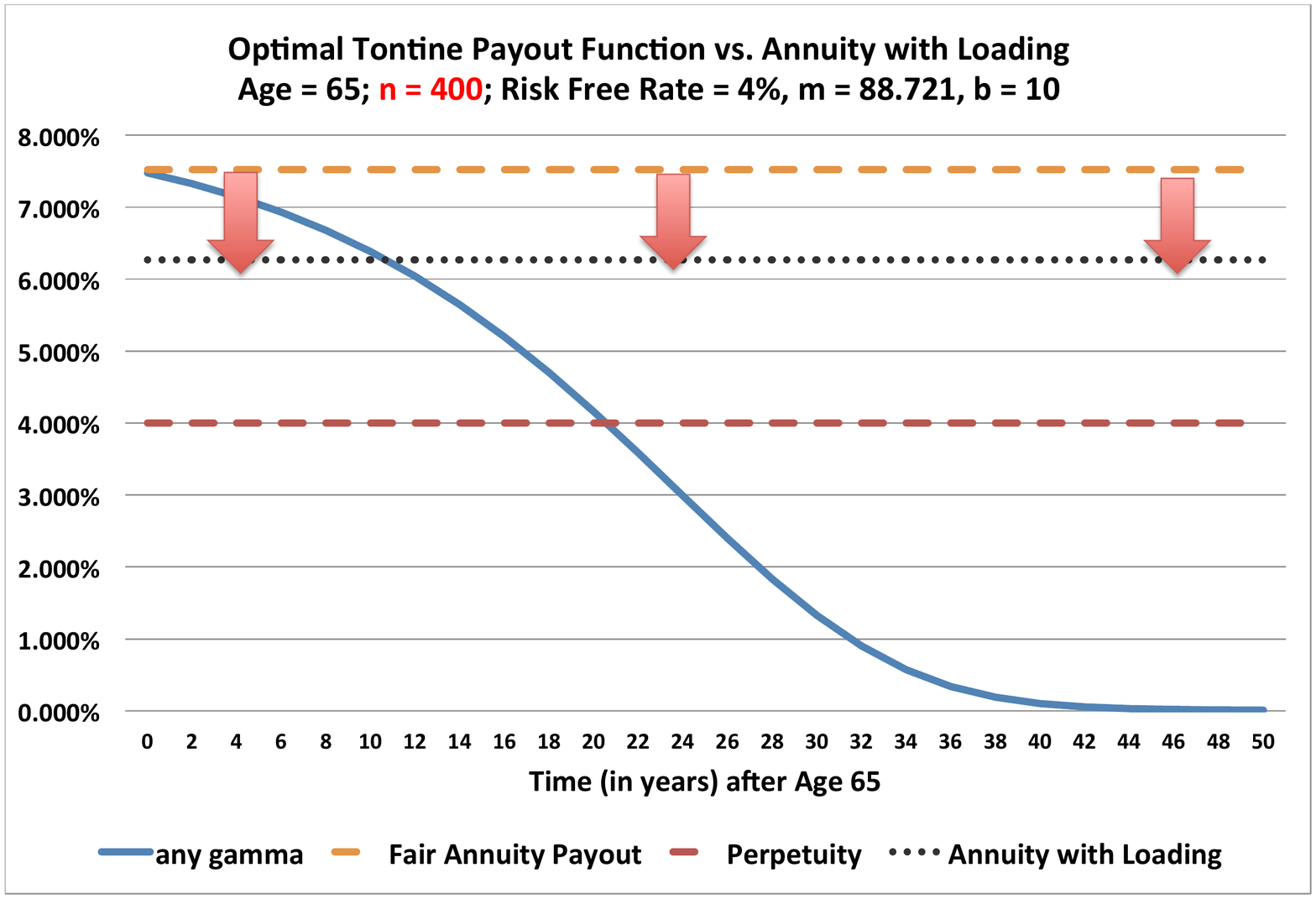} 
\caption{Shows an actuarially fair life annuity, and a natural tontine payout (LoRA $\gamma=1$). They agree at $t=0$, but the tontine pays more at time zero once loading is included. The utility of the loaded life annuity may be lower than that of the tontine, depending on the size of the loading. }
\label{fig6}
\end{center}
\end{figure}


\begin{thebibliography}{99}                                                                                               

\bibitem{A1986}
Alter, G. (1986),
How to Bet on Lives: A Guide to Life Contingent Contracts in Early Modern Europe,
\emph{Research in Economic History}, Vol. 10, pg. 1-53.

\bibitem{B2013}
Barberis, N.C. (2013),
Thirty Years of Prospect Theory in Economics: A Review and Assessment,
\emph{The Journal of Economic Perspectives}, Vol. 27(1), pg. 173-196

\bibitem {B2011}
Bellhouse, D. R. (2011),
\emph{De Moivre: Setting the Stage for Classical Probability and Its Applications}, CRC Press, New York.

\bibitem{BHYZ2013}
Bernard, C., X. He, J.-A. Yan and X. Y. Zhou (2013),
Optimal Insurance Design Under Rank-Dependent Utility, 
\emph{Mathematical Finance}, Vol. 25 (1), pg. 154--186
  
\bibitem{CT2008}
Cannon, E. and I. Tonks (2008),
\emph{Annuity Markets}, Oxford University Press, UK.

\bibitem{C2001}
Chancellor, E. (2001),
Life Long and Prosper, 
\emph{The Spectator}, 24 March.

\bibitem{CT2007}
Chung, J. and G. Tett (2007), 
Death and the Salesmen: As people live longer, pension funds struggle to keep up, which is where a new, highly profitable market will come in -- one that bets on matters of life and death,
\emph{The Financial Times}, February 24, page 26.
 
\bibitem{C2008a}
Ciecka, J.E. (2008a),
The First Mathematically Correct Life Annuity Valuation Formula,
\emph{Journal of Legal Economics}, Vol. 15(1), pg. 59-63.

\bibitem{C2008b}
Ciecka, J.E. (2008b),
Edmond Halley's Life Table and Its Uses,
\emph{Journal of Legal Economics}, Vol. 15(1), pg. 65-74.
 
\bibitem {C1972}
Cooper, R. (1972),
\emph{An Historical Analysis of the Tontine Principle with Emphasis on Tontine and Semi-Tontine Life Insurance Policies}, 
S.S. Huebner Foundation for Insurance Education, University of Pennsylvania.

\bibitem{DKZ2008}
Dai, M., Y.K. Kwok and J. Zong (2008),
Guaranteed Minimum Withdrawal Benefit in Variable Annuities,
\emph{Mathematical Finance}, Vol. 18(3), pg. 595-611

\bibitem{Dickson1967}
Dickson, P.G.M. (1967),
\emph{The Financial Revolution in England
A Study in the Development of Public Credit, 1688-1756}.
Published by Macmillan

\bibitem{DGN2013}
Donnelly, C., M. Guillen and J.P. Nielsen (2013),
Exchanging Mortality for a Cost,
\emph{Insurance: Mathematics and Economics}, Vol. 52(1), pg. 65-76. 

\bibitem{DGN2014}
Donnelly, C., M. Guill\'en and J.P. Nielsen (2014), 
Bringing cost transparency to the life annuity market, 
\emph{Insurance: Mathematics and Economics}, Vol. 56 (1), pg. 14?27.

\bibitem{D2015}
Donnelly, C. (2015)
Actuarial Fairness and Solidarity in Pooled Annuity Funds,
\emph{ASTIN Bulletin}, Vol. 45, pg. 49-74
 
\bibitem{E2007}
Elsgolc, L.D. (2007),
\emph{Calculus of Variations},
Dover Publications, Mineola, New York.

\bibitem{F1829}
Finlaison, J. (1829),
\emph{Report of John Finlaison, Actuary of the National Debt, on the Evidence and Elementary Facts on which the Tables of Life Annuities are Founded}, House of Commons, London, UK.

\bibitem{GF2000}
Gelfand, I.M. and S.V. Fomin (2000),
\emph{Calculus of Variations},
translated by R.A. Silverman,
Dover Publications, Mineola, New York.

\bibitem{G2007}
Goldsticker, R. (2007)
A Mutual Fund to Yield Annuity-like Benefits,
\emph{Financial Analysts Journal}, Vol. 63 (1), pg. 63-67.
 
\bibitem{H2003}
Hald, A. (2003),
\emph{History of Probability and Statistics and Their Applications Before 1750},
Wiley Series in Probability and Statistics, John Wiley \& Sons, New Jersey

\bibitem{HS2005}
Homer, S. and R. Sylla (2005),
\emph{A History of Interest Rates}, 4th Edition,
John Wiley and Sons, New Jersey.

\bibitem {JT1982}
Jennings, R. M. and A. P. Trout (1982),
\emph{The Tontine: From the Reign of Louis XIV to the French Revolutionary Era}, 
S.S. Huebner Foundation for Insurance Education, University of Pennsylvania, 91 pages.
 
\bibitem{JST 1988}
Jennings, R. M., D. F. Swanson, and A. P. Trout (1988)
Alexander Hamilton's Tontine Proposal,
\emph{The William and Mary Quarterly}, Vol. 45(1), pg. 107-115.
 
\bibitem{K1927}
Kopf, E. W. (1927)
The Early History of the Annuity,
\emph{Proceedings of the Casualty Actuarial Society}, Vol. 13(28), pg. 225-266.
 
\bibitem {L2003}
Lewin, C.G. (2003),
\emph{Pensions and Insurance Before 1800: A Social History}, 
Tuckwell Press, East Lothian, Scotland.

\bibitem{M2009}
McKeever, K. (2009),
A Short History of Tontines,
\emph{Fordham Journal of Corporate and Financial Law}, Vol. 15(2), pg. 491-521.

\bibitem{M2014}
Milevsky, M.A. (2014),
Portfolio Choice and Longevity Risk in the Late 17th Century: A Re-examination of the First English Tontine,
\emph{Financial History Review}, Vol. 21(3), pg. 225-258
 
\bibitem{M2015}
Milevsky, M.A. (2015), \emph{King William's Tontine: Why the Retirement Annuity of the Future Should Resemble Its Past}, Cambridge University Press, New York City.
 
\bibitem{MS2015}
Milevsky, M.A. and T.S. Salisbury (2015),
On the Choice Between Tontines and Annuities Under Stochastic and Asymmetric Mortality, \emph{manuscript in preparation}

\bibitem{Moody2013}
Moody's Investor Services (2013), \emph{European Insurers: Solvency II - Volatility of Regulatory Ratios Could Have Broad Implications For European Insurers}, retrieved May 2013 from {\tt www.moodys.com}

\bibitem {T1936}
O'Donnell, T. (1936),
\emph{History of Life Insurance in its Formative Years; compiled from approved sources}, American Conservation Co., Chicago.

\bibitem{P2007}
Pechter, K. (2007),
Possible Tontine Revival Raises Worries,
\emph{Annuity Market News}, 1 May, SourceMedia.
 
\bibitem{PVD2005}
Piggott, J., E. A. Valdez, and B. Detzel (2005),
The Simple Analytics of a Pooled Annuity Fund,
\emph{The Journal of Risk and Insurance}, Vol. 72(3), pg. 497-520.
 
\bibitem{PDHO2009}
Pitacco, E., M. Denuit, S. Haberman, and A. Olivieri (2009),
\emph{Modeling Longevity Dynamics for Pensions and Annuity Business}
Oxford University Press, UK.
 
\bibitem{P2000}
Poitras, G. (2000),
\emph{The Early History of Financial Economics: 1478-1776}, 
Edward Elgar, Cheltenham UK.

\bibitem{P2005}
Poterba, J.M. (2005),
Annuities in Early Modern Europe, in
\emph{The Origins of Value: The Financial Innovations that Created Modern Capital Markets}, 
edited by W. N. Goetzmann and K. G. Rouwenhorst. New York: Oxford University Press.

\bibitem{2011}
Promislow, S.D. (2011),
\emph{Fundamentals of Actuarial Mathematics}, 2nd Edition,
John Wiley \& Sons, United Kingdom.

\bibitem{2013}
Qiao, C and M. Sherris (2013),
Managing systematic mortality risk with group self-pooling and annuitization schemes,
\emph{Journal of Risk and Insurance}, Vol. 80(4), pg. 949-974.

\bibitem{RS1987}
Ransom, R. L. and R. Sutch (1987),
Tontine Insurance and the Armstrong Investigation: A Case of Stifled Innovation, 1868-1905,
\emph{The Journal of Economic History}, Vol. 47(2), pg. 379-390.
 
\bibitem{RW2011}
Richter, A., and F. Weber (2011),
Mortality-Indexed Annuities Managing Longevity Risk via Product Design,
\emph{North American Actuarial Journal}, Vol. 15(2), pg. 212-236.
 
\bibitem{Rg2009}
Rotemberg, J. J. (2009),
Can a Continuously-liquidating Tontine (or Mutual Inheritance Fund) Succeed where Immediate Annuities have Floundered?,
\emph{Harvard Business School: Working Paper}
 
\bibitem{Rd2009}
Rothschild, C. (2009),
Adverse selection in annuity markets: Evidence from the British Life Annuity Act of 1808,
\emph{Journal of Public Economics}, Vol. 93(5-6), pg. 776-784.
 
\bibitem{S2010}
Sabin, M.J. (2010),
Fair Tontine Annuity,
\emph{SSRN abstract \#1579932}
 
\bibitem{S2008}
Stamos, M. Z. (2008),
Optimal Consumption and Portfolio Choice for Pooled Annuity Funds,
\emph{Insurance: Mathematics and Economics}, Vol. 43 (1), pg. 56-68.

\bibitem{Steele1997}
Steele, J. M. (1997),
\emph{Probability Theory and Combinatorial Optimization}, 
CBMS-NSF Regional Conference Series in Applied Mathematics, Vol 69, SIAM, Philadelphia, Pennsylvania
 
\bibitem{T1654}
Tonti, L. (1654),
\emph{Edict of the King for the Creation of the Society of the Royal Tontine},
Translated by V. Gasseau-Dryer, published in 
\emph{History of Actuarial Science}, Volume V, 
Edited by S. Haberman and T.A. Sibbett, 
published by William Pickering, London, 1995
 
\bibitem{VPW2006}
Valdez, E. A., J. Piggott, and L. Wang (2006),
Demand and Adverse Selection in a Pooled Annuity Fund,
\emph{Insurance: Mathematics and Economics}, Vol. 39(2), pg. 251-266.
 
\bibitem{W1989}
Weir, D. R. (1989),
Tontines, Public Finance, and Revolution in France and England,
\emph{The Journal of Economic History}, Vol. 49 (1), pg. 95-124.
 
\bibitem{Y1965}
Yaari, M. E. (1965),
Uncertain Lifetime, Life Insurance and the Theory of the Consumer,
\emph{The Review of Economic Studies}, Vol. 32(2), pg. 137-150.
 
\end{thebibliography}
\end{document}